\documentclass[conference]{IEEEtran}
\usepackage{textcomp}
\usepackage[latin9]{inputenc}
\usepackage{xcolor}
\usepackage{array}
\usepackage{pifont}
\usepackage{float}
\usepackage{booktabs}
\usepackage{bbding}
\usepackage{multirow}
\usepackage{amsmath}
\usepackage{amsthm}
\usepackage{amssymb}
\usepackage{graphicx}
\PassOptionsToPackage{normalem}{ulem}
\usepackage{ulem}

\makeatletter

\providecommand{\tabularnewline}{\\}
\floatstyle{ruled}
\newfloat{algorithm}{tbp}{loa}
\providecommand{\algorithmname}{Algorithm}
\floatname{algorithm}{\protect\algorithmname}
\providecolor{lyxadded}{rgb}{0,0,1}
\providecolor{lyxdeleted}{rgb}{1,0,0}
\DeclareRobustCommand{\mklyxadded}[1]{\textcolor{lyxadded}\bgroup#1\egroup}
\DeclareRobustCommand{\mklyxdeleted}[1]{\textcolor{lyxdeleted}\bgroup\mklyxsout{#1}\egroup}
\DeclareRobustCommand{\mklyxsout}[1]{\ifx\\#1\else\sout{#1}\fi}

\theoremstyle{plain}
\newtheorem{thm}{\protect\theoremname}

\IEEEoverridecommandlockouts

\usepackage{cite}
\usepackage{amsfonts}
\usepackage{amsthm}

\usepackage{textcomp}
\usepackage{xcolor}
\def\BibTeX{{\rm B\kern-.05em{\sc i\kern-.025em b}\kern-.08em
   T\kern-.1667em\lower.7ex\hbox{E}\kern-.125emX}}

\providecommand{\tabularnewline}{\\}
\floatstyle{ruled}
\newfloat{algorithm}{tbp}{loa}
\providecommand{\algorithmname}{Algorithm}
\floatname{algorithm}{\protect\algorithmname}

\usepackage{algorithm}
\usepackage{algpseudocode}
\algnewcommand{\Linecomment}[1]{\Statex \(\triangleright\) #1}
\algrenewcommand\algorithmicrequire{\textbf{Input:}}
\algrenewcommand\algorithmicensure{\textbf{Output:}}
\usepackage[justification=centering]{caption}

\theoremstyle{plain}

\newtheorem{lemma}{Lemma}
\newtheorem{Corollary}{Corollary}

\makeatother

\providecommand{\theoremname}{Theorem}

\begin{document}
\title{Scheduling Coflows in Multi-Core OCS Networks with Performance Guarantee}
\author{{\normalsize Xin Wang$^{a}$, Hong Shen$^{a}$, Hui Tian$^{b}$, Dong
Wang$^{a}$}\\
 \\
 {\normalsize$^{a}$School of Engineering and Technology, Central
Queensland University, Australia}\\
 {\normalsize$^{b}$School of Information and Communication Technology,
Griffith University, Australia}}
\maketitle
\begin{abstract}
The coflow abstraction captures application-level communication patterns
and enables coordinated scheduling of parallel flows to reduce job
completion times in distributed systems. Modern data center networks
(DCNs) are employing multiple independent optical circuit switching
(OCS) cores operating concurrently to meet the massive bandwidth demands
of application jobs. However, existing coflow scheduling research
primarily focuses on the single-core setting, while studies of multi-core
fabrics have largely considered electrical packet switching (EPS)
networks.

To address this gap, this paper studies the coflow scheduling problem
in multi-core OCS networks under the not-all-stop reconfiguration
model, in which the reconfiguration of one circuit does not interrupt
other circuits. The challenges stem from two aspects: (i) cross-core
coupling induced by traffic assignment across heterogeneous cores;
and (ii) per-core OCS scheduling constraints, namely \textit{port
exclusivity} and \textit{reconfiguration delay}. We propose an approximation
algorithm that jointly integrates cross-core flow assignment and per-core
circuit scheduling to minimize the total weighted coflow completion
time (CCT) and establish a provable worst-case performance guarantee.
Furthermore, our algorithm framework can be applied to the multi-core
EPS scenario with a corresponding approximation guarantee for packet-switched
fabrics. Trace-driven simulations using real Facebook workloads demonstrate
that our algorithm can reduce the total weighted CCT and tail CCT.
\end{abstract}


\section{Introduction}
In systems such as MapReduce~\cite{mapreduce}, Spark~\cite{spark},
and Dryad~\cite{dryad}, job progress depends on the completion
of parallel flows within communication stages. During the
MapReduce shuffle, for example, reducers must receive all required
intermediate outputs from map workers before proceeding, so the
stage finishes only after its last flow completes.
The coflow abstraction~\cite{networking} groups these related
flows for coordinated scheduling, with coflow completion time
(CCT) defined as the completion time of the last flow.
This metric captures the communication dependency and reflects
job progress more directly than individual flow completion
time (FCT) alone.

Coflow scheduling has been studied primarily in single-core
networks. In electrical packet switching (EPS) models, the
data center network (DCN) is commonly represented as a
non-blocking fabric with full bisection bandwidth, allowing
schedulers to focus on port bandwidth constraints
\cite{literature4,literature6,Baraat,decentralized1,Aalo,CODA,wang2023online,wang2023efficient,literature9,improved,iaoa}.
As traffic volumes grow, meeting bandwidth demands within
cost and energy budgets becomes more difficult, motivating
the use of optical circuit switching (OCS).
By providing dedicated high-bandwidth circuits for bulk flows,
OCS introduces a different scheduling model that has also
been examined in single-core networks
\cite{omco,reco,regularization,sunflow,zhang2020minimizing}.
In addition, hybrid EPS-OCS designs have been studied to
coordinate packet and circuit switching within the same
network~\cite{wang2024scheduling,co-scheduler,wang2025optimal,ONS}.

However, a single switching core does not capture the parallel
resources available in modern DCNs. Parallel designs expand
aggregate bandwidth by operating multiple cores concurrently
\cite{cisco2016_40g,cisco2016gci}.
Cores from different generations may coexist with different
capacities, forming heterogeneous parallel networks (HPNs)
that serve a common set of hosts~\cite{huang2020weaver}.
This architecture has motivated coflow scheduling algorithms
for multi-core EPS networks
\cite{huang2020weaver,chen2023efficient}.
Parallel optical switching is also used in Google's Jupiter
architecture, whose evolution from a Clos fabric toward
a direct-connect topology incorporates multiple OCS resources
for topology reconfiguration and incremental capacity
expansion~\cite{poutievski2022jupiter}.
While these designs support flexible expansion, coflow
scheduling across multiple OCS cores has received limited
attention and requires coordinating traffic assignment
with circuit scheduling.

Scheduling across OCS cores must account for both
\textit{port exclusivity} and \textit{reconfiguration delay}.
Each ingress or egress port supports at most one active circuit,
and establishing a circuit incurs a delay $\delta$, typically
ranging from hundreds of microseconds to milliseconds.
The reconfiguration model determines which transmissions
are interrupted: all-stop reconfiguration pauses the entire
affected core, whereas not-all-stop reconfiguration allows
flows on unchanged circuits to continue
(see Subsection~\ref{subsec:Reconfiguration-Mechanism}).
We adopt the not-all-stop model with non-preemptive flow
transmission, so each flow runs to completion once started.
Scheduling must therefore coordinate circuit updates with
port availability while other circuits remain active.

Flow assignment across cores determines both transmission load
and reconfiguration workload, while circuit scheduling determines
when each assigned flow can be served. These decisions are coupled,
requiring more than either single-core OCS scheduling or traffic
distribution across EPS cores alone.
We develop an approximation algorithm that coordinates flow
assignment and per-core circuit scheduling to minimize total
weighted CCT under not-all-stop reconfiguration.
Our analysis establishes a worst-case approximation guarantee
and a refined deterministic bound incorporating coflow weights
and workloads.
The framework also extends to multi-core EPS networks by
setting the reconfiguration delay to zero and using the
corresponding EPS lower bounds, with an approximation
guarantee for that setting.

Section~\ref{sec:Related-Work} discusses prior studies.
The system model and scheduling problem are defined in
Section~\ref{sec:Model-and-Problem}, followed by the algorithm
and its performance guarantees in
Section~\ref{sec:Multiple-Coflow-Scheduling}.
Section~\ref{sec:Experimental-Evaluations} evaluates the
algorithm using the Facebook trace, and
Section~\ref{sec:Conclusions} presents the conclusions.

\section{Related Work}
\label{sec:Related-Work}
We organize prior studies by switching technology and core
configuration: single-core EPS, single-core OCS, hybrid
EPS-OCS, multi-core EPS, and multi-core OCS networks.
For OCS-based designs, we further distinguish the results
under all-stop and not-all-stop reconfiguration.

\subsection{Single-Core EPS Networks}
Early coflow schedulers mainly address bandwidth allocation
and scheduling priorities. Varys~\cite{literature6} combines
smallest-effective-bottleneck-first (SEBF) ordering with
minimum-allocation-for-desired-duration (MADD) bandwidth
allocation to reduce CCT.
For decentralized scheduling, Baraat~\cite{Baraat} uses
task-aware dynamic multiplexing to mitigate head-of-line
blocking, while D-CAS~\cite{decentralized1} supports
distributed decisions and preemption.
Aalo~\cite{Aalo} applies discretized coflow-aware
least-attained service (D-CLAS) without requiring coflow
sizes in advance, and CODA~\cite{CODA} uses machine
learning to identify coflows from individual flows.
Wang \textit{et al.}~\cite{wang2023online} use deep
reinforcement learning (DRL) for online coflow scheduling
in multi-stage jobs, aiming to reduce weighted job
completion time. Their subsequent
work~\cite{wang2023efficient} combines limited multiplexing
with DRL to improve average weighted CCT and fairness.
Rapier~\cite{theoretical1} additionally considers network
routing, coordinating path selection with coflow scheduling.

Theoretical work establishes approximation guarantees
under specific workload assumptions.
For total weighted CCT with zero release times,
Qiu \textit{et al.}~\cite{literature9} give a deterministic
$64/3$-approximation.
Khuller and Purohit~\cite{literature10} obtain a factor
of $8$ by relating coflow scheduling to concurrent
open-shop scheduling, while Shafiee and
Ghaderi~\cite{improved} obtain a factor of $4$ using
linear programming (LP), both for coflows available
at time zero.
Under this release-time assumption, Wang
\textit{et al.}~\cite{iaoa} report a factor of $2$
for their idealized offline weighted formulation
without solving an LP.
For the matroid coflow model, Im
\textit{et al.}~\cite{theoretical4} obtain factors
of $2$ for unit-length tasks and $2+\epsilon$
for arbitrary task lengths, for any fixed $\epsilon>0$.
Shafiee and Ghaderi~\cite{theoretical5} consider
dependencies between coflows within multi-stage jobs
and develop polynomial-time approximation algorithms
for minimizing total weighted job completion time (JCT).

\subsection{Single-Core OCS Networks}
Under all-stop reconfiguration, OMCO~\cite{omco} schedules
multiple coflows online using Birkhoff-von Neumann (BvN)
decomposition~\cite{BvN}.
Reco-Sin~\cite{reco} combines regularization with BvN
decomposition to obtain a factor of $2$ for single-coflow
scheduling. For multiple coflows available at time zero,
Reco-Mul+~\cite{regularization} achieves an $8M$ bound
for unweighted total CCT, where $M$ denotes the coflow
count. Its guarantee for weighted CCT also depends on
ratios of coflow loads and sparsity.

For not-all-stop reconfiguration, Sunflow~\cite{sunflow}
offers a constant-factor approximation for a single
coflow and a heuristic for multiple coflows.
GOS~\cite{zhang2020minimizing} addresses the multi-coflow
case, with an offline analysis establishing a factor
of $4$ for average CCT. 

\subsection{Hybrid EPS-OCS Networks}
Hybrid scheduling coordinates the use of packet and circuit
resources. Under all-stop reconfiguration, Wang
\textit{et al.}~\cite{wang2024scheduling} provide approximation
algorithms for both single and multiple coflows.
Co-scheduler~\cite{co-scheduler} incorporates coflow structure
into task placement and scheduling, while ONS~\cite{ONS}
uses network state to guide online coflow priorities
and flow-path selection.

For not-all-stop reconfiguration, Wang
\textit{et al.}~\cite{wang2026coflow} obtain an
$O(\tau_D)$ approximation for single-coflow scheduling,
where $\tau_D\leq N$ denotes the maximum number of nonzero
entries in a row or column of the $N\times N$ demand
matrix $D$. They also adapt selected pure-EPS schedulers
to the hybrid setting for multiple coflows, yielding
a heuristic without the single-coflow approximation guarantee.

\subsection{Multi-Core EPS Networks}
Parallel EPS networks require traffic assignment across cores
as well as scheduling within each core. In a heterogeneous
parallel network, Weaver~\cite{huang2020weaver} provides an
$O(H)$ guarantee for a single coflow,
where $H$ is the number of EPS cores. Chen~\cite{chen2023scheduling}
studies multiple coflows with arbitrary release times in
heterogeneous parallel networks. For the flow-level assignment
variant, in which different flows of a coflow may use different
cores, the approximation factor is
$O\!\left(\frac{\log H}{\log\log H}\right)$.

For identical parallel networks, Chen~\cite{chen2023efficient}
also studies the coflow-level assignment variant, which assigns
all flows of a coflow to the same core. The corresponding
approximation factors are $4H+1$ for arbitrary release times
and $4H$ when all coflows are available initially. These
packet-switching results address parallel resource allocation
without the circuit reconfiguration constraints of OCS networks.

\subsection{Multi-Core OCS Networks}
With multiple heterogeneous OCS cores, assigning traffic to a core
determines both its transmission load and its reconfiguration
workload. Each core operates independently and must satisfy its
own port-exclusivity constraints. Under not-all-stop
reconfiguration, circuit scheduling must additionally account for
the temporary unavailability of ports involved in circuit changes.
Consequently, traffic assignment across cores and circuit
scheduling within each core must be considered jointly.

The studies reviewed above address single-core EPS and OCS
networks, hybrid EPS-OCS networks, and multi-core EPS networks.
Coflow scheduling in multi-core OCS networks has received
comparatively limited attention. This paper addresses this setting
under not-all-stop reconfiguration and first develops an approximation
algorithm with a provable worst-case guarantee for minimizing
total weighted CCT.

\section{System Model and Problem Formulation\label{sec:Model-and-Problem}}

This section formulates multi-coflow scheduling in heterogeneous
parallel networks (HPNs). We describe the network, coflow demands,
and reconfiguration constraints before defining the optimization
problem and establishing its computational hardness.
The notation is summarized in Table~\ref{tab:notation1}.

\begin{table*}[t]
\centering \caption{Mathematical Notation}
\label{tab:notation1} {\footnotesize{}}{\footnotesize{}%
\begin{tabular}{@{}p{0.12\textwidth}p{0.75\textwidth}@{}}
\toprule 
{\footnotesize{}{}Notation} & {\footnotesize{}{}Definition}\tabularnewline
\midrule 
{\footnotesize{}{}$\mathcal{M}$} & {\footnotesize{}{}Set of coflows}\tabularnewline
{\footnotesize{}{}$M$} & {\footnotesize{}{}Number of coflows, i.e., $M=\left|\mathcal{M}\right|$}\tabularnewline
{\footnotesize{}{}$\mathcal{K}$} & {\footnotesize{}{}Set of parallel OCS cores}\tabularnewline
{\footnotesize{}{}$K$} & {\footnotesize{}{}Number of OCS cores, i.e., $K=\left|\mathcal{K}\right|$}\tabularnewline
{\footnotesize{}{}$N$} & {\footnotesize{}{}Number of ingress/egress ports}\tabularnewline
{\footnotesize{}{}$\mathcal{I}$, $\mathcal{J}$} & {\footnotesize{}{}Ingress- and egress-port index sets, respectively,
where $|\mathcal{I}|=|\mathcal{J}|=N$}\tabularnewline
{\footnotesize{}{}$\mathcal{P}$} & {\footnotesize{}{}Disjoint union $\mathcal{I}\mathbin{\dot{\cup}}\mathcal{J}$
of the two port sets}\tabularnewline
{\footnotesize{}{}$p$} & {\footnotesize{}{}Any port in $\mathcal{P}$}\tabularnewline
{\footnotesize{}{}$C_{m}$} & {\footnotesize{}{}The $m$-th coflow, where $1\leq m\leq M$}\tabularnewline
{\footnotesize{}{}$\mathcal{F}_{m}$} & {\footnotesize{}{}Set of flows in $C_{m}$}\tabularnewline
{\footnotesize{}{}$D_{m}$} & {\footnotesize{}{}Demand matrix of $C_{m}$}\tabularnewline
{\footnotesize{}{}$D_{m}^{k}$} & {\footnotesize{}{}Traffic matrix assigned to core $k$, where $\sum_{k=1}^{K}D_{m}^{k}=D_{m}$}\tabularnewline
{\footnotesize{}{}$f_{m}\left(i,j\right)$} & {\footnotesize{}{}Flow from ingress port $i$ to egress port $j$
of $C_{m}$}\tabularnewline
{\footnotesize{}{}$f_{m}^{k}\left(i,j\right)$} & {\footnotesize{}{}Subflow of $f_{m}\left(i,j\right)$ transmitted
on core $k$}\tabularnewline
{\footnotesize{}{}$t_{m}^{k}\left(i,j\right)$} & {\footnotesize{}{}Time at which circuit reconfiguration for $f_{m}^{k}\left(i,j\right)$
begins on core $k$}\tabularnewline
{\footnotesize{}{}$d_{m}\left(i,j\right)$,$d_{m}^{k}\left(i,j\right)$} & {\footnotesize{}{}Data size of $f_{m}\left(i,j\right)$ and $f_{m}^{k}\left(i,j\right)$,
respectively}\tabularnewline
{\footnotesize{}{}$\rho_{m,p},\rho_{m,p}^{k}$} & {\footnotesize{}{}Traffic load incident to port $p$ in $D_{m}$
and $D_{m}^{k}$, respectively}\tabularnewline
{\footnotesize{}{}$\rho_{m}$} & {\footnotesize{}{}Maximum traffic load incident to any port of $C_{m}$}\tabularnewline
{\footnotesize{}{}$\rho_{m}^{k}$} & {\footnotesize{}{}Maximum traffic load incident to any port on core
$k$ for $C_{m}$}\tabularnewline
{\footnotesize{}{}$\tau_{m,p}$} & {\footnotesize{}{}Number of nonzero flows incident to port $p$
in $D_{m}$}\tabularnewline
{\footnotesize{}{}$\tau_{m,p}^{k}$} & {\footnotesize{}{}Number of nonzero subflows incident to port $p$
in $D_{m}^{k}$}\tabularnewline
{\footnotesize{}{}$\tau_{m}$} & {\footnotesize{}{}Maximum number of nonzero flows incident to any
port in $D_{m}$}\tabularnewline
{\footnotesize{}{}$\tau_{m}^{k}$} & {\footnotesize{}{}Maximum number of nonzero subflows incident to
any port in $D_{m}^{k}$}\tabularnewline
{\footnotesize{}{}$\pi$} & {\footnotesize{}{}Global coflow priority order}\tabularnewline
{\footnotesize{}{}$\rho_{1:m,p}$} & {\footnotesize{}{}Cumulative load incident to port $p$ in the first
$m$ coflows under $\pi$, i.e., $\sum_{\ell=1}^{m}\rho_{\pi(\ell),p}$}\tabularnewline
{\footnotesize{}{}$\rho_{1:m,p}^{k}$} & {\footnotesize{}{}Cumulative load incident to port $p$ on core
$k$ in the first $m$ coflows under $\pi$, i.e., $\sum_{\ell=1}^{m}\rho_{\pi(\ell),p}^{k}$}\tabularnewline
{\footnotesize{}{}$\rho_{1:m}$, $\rho_{1:m}^{k}$} & {\footnotesize{}{}Maximum cumulative port loads $\max_{p\in\mathcal{P}}\rho_{1:m,p}$
and $\max_{p\in\mathcal{P}}\rho_{1:m,p}^{k}$, respectively}\tabularnewline
{\footnotesize{}{}$\tau_{1:m,p}$} & {\footnotesize{}{}Cumulative number of nonzero flows incident to
port $p$ in the first $m$ coflows under $\pi$, i.e., $\sum_{\ell=1}^{m}\tau_{\pi(\ell),p}$}\tabularnewline
{\footnotesize{}{}$\tau_{1:m,p}^{k}$} & {\footnotesize{}{}Cumulative number of nonzero subflows incident
to port $p$ on core $k$ in the first $m$ coflows under $\pi$,
i.e., $\sum_{\ell=1}^{m}\tau_{\pi(\ell),p}^{k}$}\tabularnewline
{\footnotesize{}{}$\tau_{1:m}$, $\tau_{1:m}^{k}$} & {\footnotesize{}{}Maximum cumulative numbers of nonzero flows and
subflows incident to any port, respectively, i.e., $\max_{p\in\mathcal{P}}\tau_{1:m,p}$
and $\max_{p\in\mathcal{P}}\tau_{1:m,p}^{k}$}\tabularnewline
{\footnotesize{}{}$r^{k}$, $r_{\max}$} & {\footnotesize{}{}Per-port transmission rate of core $k$ and the
maximum rate $\max_{k}r^{k}$}\tabularnewline
{\footnotesize{}{}$R$} & {\footnotesize{}{} Aggregate port transmission rate, i.e., $R=\sum_{k=1}^{K}r^{k}$}\tabularnewline
{\footnotesize{}{}$w_{m}$} & {\footnotesize{}{}Weight of $C_{m}$}\tabularnewline
{\footnotesize{}{}$\delta$} & {\footnotesize{}{} Reconfiguration delay}\tabularnewline
{\footnotesize{}{}$T_{m}^{k}(i,j)$} & {\footnotesize{}{}Completion time of subflow $f_{m}^{k}(i,j)$ on
core $k$}\tabularnewline
{\footnotesize{}{}$T_{m}^{k}$} & {\footnotesize{}{}Completion time of the portion of coflow $C_{m}$
assigned to core $k$}\tabularnewline
{\footnotesize{}{}$T_{m}$} & {\footnotesize{}{}Completion time of coflow $C_{m}$, i.e., $T_{m}=\max_{k}T_{m}^{k}$}\tabularnewline
\bottomrule
\end{tabular}}{\footnotesize{}{}}
\end{table*}

\subsection{Network Architecture}
The network in Fig.~\ref{fig:HPN} connects source servers
$\{s_1,\ldots,s_N\}$ to destination servers
$\{d_1,\ldots,d_N\}$ through $K$ parallel OCS cores.
Every source has one uplink to each core, and every
destination has one downlink from each core.
Thus, each server has access to all $K$ switching cores.

\begin{figure}[h]
	\centering
	\includegraphics[width=0.46\textwidth,height=5.2cm]{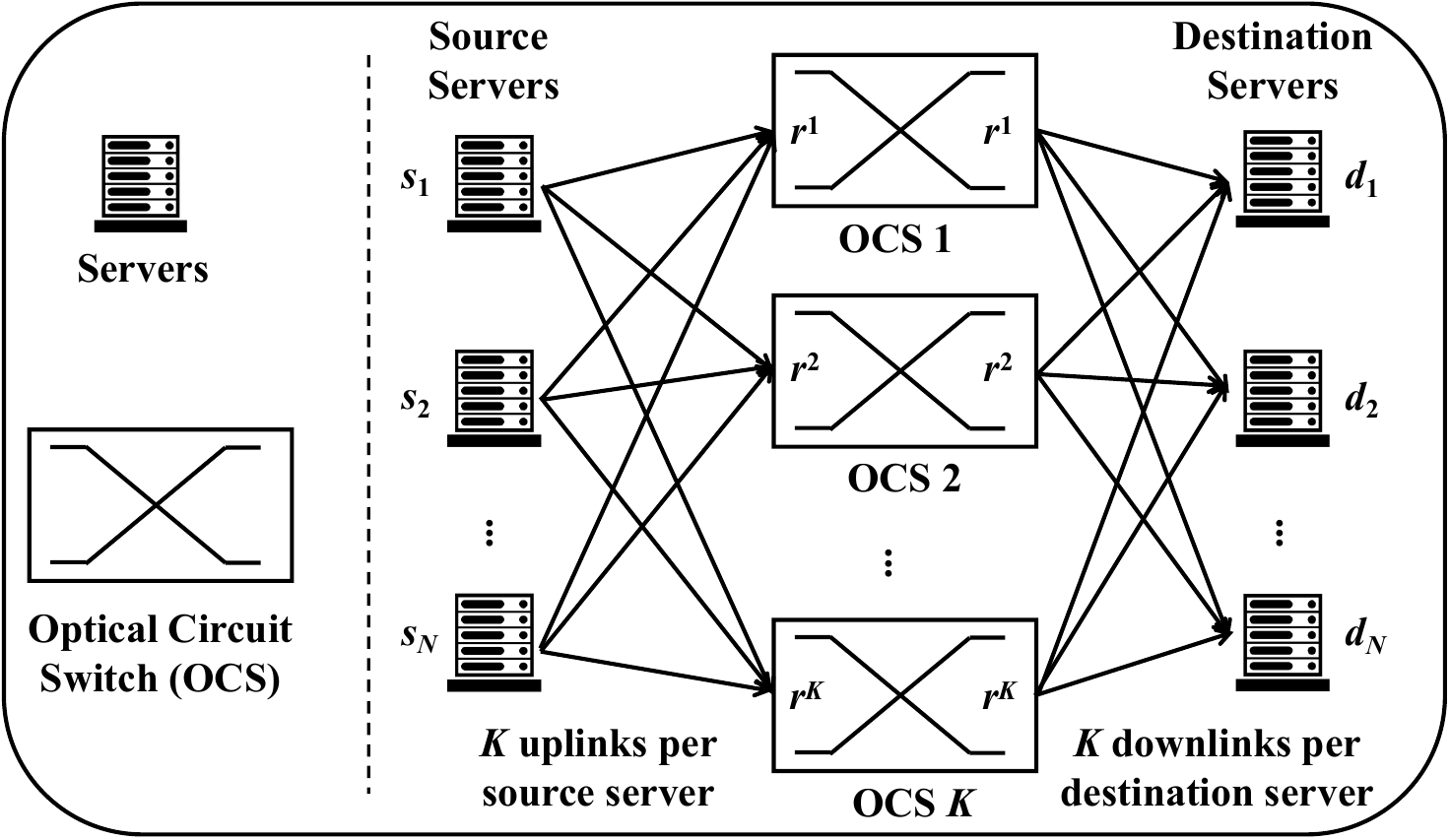}
	\caption{Heterogeneous Multi-Core DCN Architecture}
	\label{fig:HPN}
\end{figure}

Core $k\in\{1,\ldots,K\}$ is a non-blocking
$N\times N$ optical circuit switch with per-port
transmission rate $r^k$. The rates may differ across
cores, representing heterogeneous link capacities.
Traffic assignment determines which cores carry the
flows, while each core maintains its own circuit schedule
and operates independently.

For each core, we distinguish the ingress port set
$\mathcal{I}$ from the egress port set $\mathcal{J}$,
with $|\mathcal{I}|=|\mathcal{J}|=N$.
Their disjoint union is
$\mathcal{P}=\mathcal{I}\mathbin{\dot{\cup}}\mathcal{J}$.
Both sets use indices $\{1,\ldots,N\}$:
ingress port $i\in\mathcal{I}$ connects to source $s_i$,
and egress port $j\in\mathcal{J}$ connects to
destination $d_j$.

\subsection{Traffic Abstraction}
We describe the workload using the coflow abstraction~\cite{networking},
which groups the parallel data flows associated with one
communication stage of an application. The stage completes only
after every flow in the group has finished.

Let $\mathcal{M}=\{C_1,\ldots,C_M\}$ be the set of coflows,
and let $\mathcal{F}_m$ contain the flows of $C_m$.
We represent the demands of $C_m$ by the matrix
$D_m=[d_m(i,j)]_{1\leq i,j\leq N}$, where $d_m(i,j)$
specifies the data volume sent from ingress port $i$ to
egress port $j$. The flow from ingress port $i$ to egress port $j$ in
coflow $C_m$ is denoted by $f_m(i,j)\in\mathcal{F}_m$.

\subsection{Reconfiguration Mechanism}
\label{subsec:Reconfiguration-Mechanism}
Within each OCS core, an ingress port can connect to at most
one egress port at a time. The active circuits
therefore form a matching on the bipartite graph of ingress
and egress ports. Establishing a circuit requires a fixed
reconfiguration delay $\delta$, during which the corresponding
ingress and egress ports cannot transmit data.

\begin{figure}[h]
	\centering
	\includegraphics[clip,width=0.48\textwidth,totalheight=2.6cm,keepaspectratio,height=5.2cm]{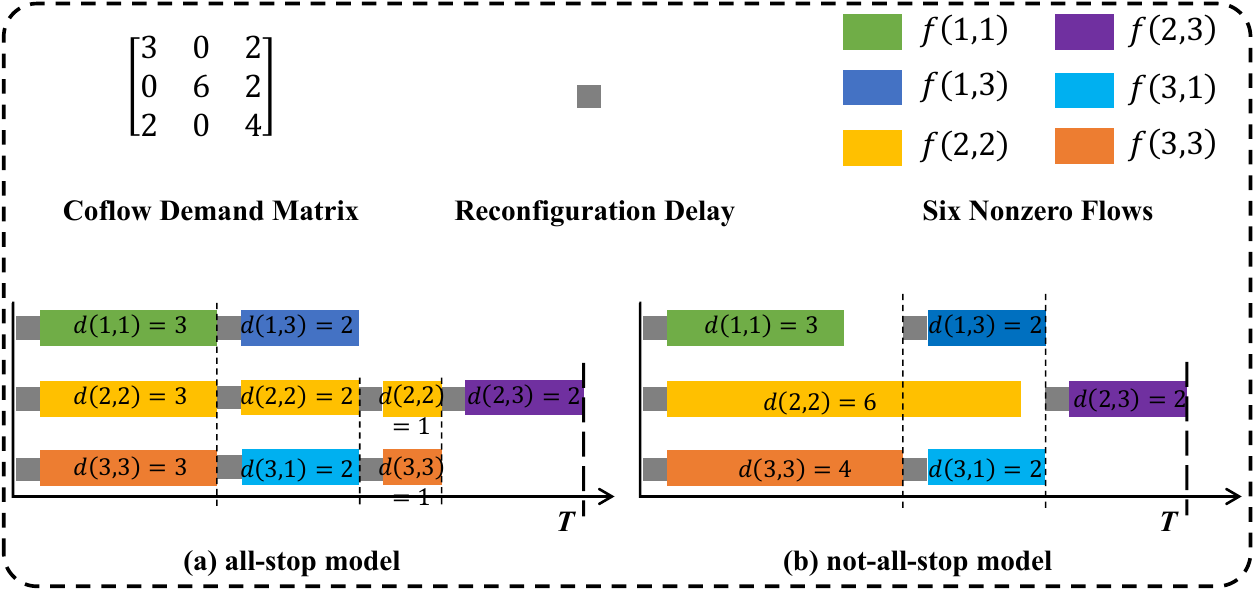}
	\caption{Circuit Reconfiguration Models in an OCS Core}
	\label{fig:reconfig}
\end{figure}

Fig.~\ref{fig:reconfig} illustrates two reconfiguration models.
In the all-stop model, every configuration change suspends
all ongoing transmissions in that core, including those
on unchanged circuits. In the not-all-stop model, only
the ports involved in the circuit update are interrupted,
while flows on unaffected circuits continue transmitting.
Circuit updates can therefore proceed asynchronously
without pausing the entire core.

We adopt the not-all-stop model and require non-preemptive
flow transmission: once a flow begins transmitting,
it continues until completion. Scheduling must coordinate
circuit updates with the availability of the corresponding
ingress and egress ports while allowing flows on other
circuits to remain active. We further assume that circuits are not reused between
consecutive coflows. Every scheduled nonzero flow incurs
its own reconfiguration delay $\delta$, even when it uses
the same ingress--egress pair as a flow from another coflow.

\subsection{Problem Definition}

All $M$ coflows in $\mathcal{M}=\{C_1,\ldots,C_M\}$
are available at time zero. The input consists of their
nonzero demand matrices $D_m$, positive weights $w_m$,
the per-port rates $r^k>0$ of the $K$ OCS cores, and
the reconfiguration delay $\delta>0$. We seek a schedule
under the not-all-stop model that minimizes
$\sum_{m=1}^{M}w_mT_m$, where $T_m$ is the completion
time of coflow $C_m$.

Scheduling requires assigning the demand to cores and
determining when the flows are served.
Let $D_m^k=[d_m^k(i,j)]_{1\leq i,j\leq N}$ record
the demand of $C_m$ assigned to core $k$.
Every entry must satisfy $d_m^k(i,j)\geq0$ and
$\sum_{k=1}^{K}d_m^k(i,j)=d_m(i,j)$, so that
$\sum_{k=1}^{K}D_m^k=D_m$.
A permutation $\pi$ of $\{1,\ldots,M\}$ provides
the common coflow priority order used by the assignment
and circuit scheduling rules.

For each positive entry $d_m^k(i,j)$, let
$t_m^k(i,j)\geq0$ denote the time at which reconfiguration
begins for the assigned subflow $f_m^k(i,j)$.
Its transmission begins at $t_m^k(i,j)+\delta$ and
finishes at
$T_m^k(i,j)=t_m^k(i,j)+\delta+d_m^k(i,j)/r^k$.
The selected times must respect the port-exclusivity
and non-preemption requirements described in
Subsection~\ref{subsec:Reconfiguration-Mechanism}.
The schedule for $C_m$ on core $k$ is represented by
$S_m^k=\{(i,j,t_m^k(i,j)):d_m^k(i,j)>0\}$.

A coflow finishes only when all of its assigned subflows
have completed. Accordingly, for $D_m^k\neq\mathbf{0}$,
we define
$T_m^k=\max_{i,j:d_m^k(i,j)>0}T_m^k(i,j)$;
if $D_m^k=\mathbf{0}$, we set $T_m^k=0$.
The CCT is then $T_m=\max_k T_m^k$.
The optimization problem is to choose the priority order,
demand assignments, and feasible circuit schedules
that minimize the total weighted CCT.

\subsection{Hardness Analysis}

The non-preemptive single-coflow scheduling problem on a single OCS
core under the not-all-stop reconfiguration model is \textit{NP-hard}
\cite{sunflow}. The problem studied in this paper contains it as
the special case $K=1$ and $M=1$. Therefore, the multi-core, multi-coflow
scheduling problem considered here is also \textit{NP-hard}.

\section{Multi-Core Coflow Scheduling\label{sec:Multiple-Coflow-Scheduling}}

In this section, we develop an approximation algorithm for multi-coflow
scheduling in multi-core OCS networks under the not-all-stop reconfiguration
model and establish its approximation guarantee. We first derive per-core
and global lower bounds on the coflow completion time (CCT). We then
present the proposed scheduling algorithm and analyze its approximation
ratio.

\subsection{Derivation of the Lower Bound}

Consider a coflow $C_{m}$ with demand matrix $D_{m}=\big[d_{m}\left(i,j\right)\big]_{1\le i,j\le N}$.
Define the traffic loads incident to ingress port $i\in\mathcal{I}$
and egress port $j\in\mathcal{J}$ as $\rho_{m,i}=\sum_{j=1}^{N}d_{m}\left(i,j\right)$
and $\rho_{m,j}=\sum_{i=1}^{N}d_{m}\left(i,j\right)$, respectively.
Similarly, define $\tau_{m,i}=\sum_{j=1}^{N}\mathbf{1}[d_{m}(i,j)>0]$
and $\tau_{m,j}=\sum_{i=1}^{N}\mathbf{1}[d_{m}(i,j)>0]$, where $\mathbf{1}[\cdot]$
is the indicator function. For $p\in\mathcal{P}=\mathcal{I}\mathbin{\dot{\cup}}\mathcal{J}$,
let $\rho_{m,p}$ and $\tau_{m,p}$ denote the corresponding traffic
load and number of nonzero flows incident to port $p$, respectively.
We then define $\rho_{m}=\max_{p\in\mathcal{P}}\rho_{m,p}$ and $\tau_{m}=\max_{p\in\mathcal{P}}\tau_{m,p}$.
In a $K$-core OCS network, let $r^{k}$ denote the per-port transmission
rate of core $k$, and let $R=\sum_{k=1}^{K}r^{k}$ denote the aggregated
port transmission rate across all cores.

\subsubsection{\textit{Per-Core Lower Bound}}

Let $\{D_{m}^{k}\}_{k=1}^{K}$ be any assignment satisfying $D_{m}=\sum_{k=1}^{K}D_{m}^{k}$,
where $D_{m}^{k}=\big[d_{m}^{k}\left(i,j\right)\big]_{1\le i,j\le N}$
denotes the portion of $D_{m}$ assigned to core $k$. For each core
$k$, define the traffic loads incident to ingress port $i\in\mathcal{I}$
and egress port $j\in\mathcal{J}$ as $\rho_{m,i}^{k}=\sum_{j=1}^{N}d_{m}^{k}\left(i,j\right)$
and $\rho_{m,j}^{k}=\sum_{i=1}^{N}d_{m}^{k}\left(i,j\right)$, respectively.
For $p\in\mathcal{P}$, let $\rho_{m,p}^{k}$ denote the corresponding
traffic load incident to port $p$ on core $k$. The maximum port
load on core $k$ is $\rho_{m}^{k}=\max_{p\in\mathcal{P}}\rho_{m,p}^{k}=\max\left\{ \max_{i\in\mathcal{I}}\rho_{m,i}^{k},\max_{j\in\mathcal{J}}\rho_{m,j}^{k}\right\} $.
Similarly, define $\tau_{m,i}^{k}=\sum_{j=1}^{N}\mathbf{1}\left[d_{m}^{k}\left(i,j\right)>0\right]$
and $\tau_{m,j}^{k}=\sum_{i=1}^{N}\mathbf{1}\left[d_{m}^{k}\left(i,j\right)>0\right]$.
For $p\in\mathcal{P}$, let $\tau_{m,p}^{k}$ denote the number of
nonzero subflows in $D_{m}^{k}$ incident to port $p$. Accordingly,
$\tau_{m}^{k}=\max_{p\in\mathcal{P}}\tau_{m,p}^{k}$.

For a given assignment $\left\{ D_{m}^{k}\right\} _{k=1}^{K}$, define
$T_{\textrm{LB}}^{k}\left(\cdot\right)$ as the CCT lower-bound function
of the traffic assigned to core $k$. For any nonzero demand matrix
$D_{m}^{k}\neq\mathbf{0}_{N\times N}$, the per-core lower bound is
given by 
\begin{equation}
T_{\textrm{LB}}^{k}\left(D_{m}^{k}\right)=\max\left\{ \max_{1\le i\le N}L_{m,i}^{k},\max_{1\le j\le N}L_{m,j}^{k}\right\} ,\label{eq:per-core-lb}
\end{equation}
where $L_{m,i}^{k}=\frac{\rho_{m,i}^{k}}{r^{k}}+\tau_{m,i}^{k}\delta$
and $L_{m,j}^{k}=\frac{\rho_{m,j}^{k}}{r^{k}}+\tau_{m,j}^{k}\delta$.

The per-core lower bound is derived from the \textit{port exclusivity}
and \textit{reconfiguration delay} constraints. Since each port can
participate in at most one circuit at a time, transmitting all subflows
incident to ingress port $i$ requires at least $\rho_{m,i}^{k}/r^{k}$
time. Under our model, every nonzero subflow is scheduled as a separate
circuit transmission and incurs one reconfiguration interval of length
$\delta$. The reconfiguration intervals of subflows incident to the
same port cannot overlap. Therefore, the reconfigurations require
at least $\tau_{m,i}^{k}\delta$ time on ingress port $i$, and the
same bound applies to each egress port $j$. Taking the maximum over
all ports yields $T_{\mathrm{LB}}^{k}\left(D_{m}^{k}\right)$.

\subsubsection{\textit{Global Lower Bound}}

Note that $T_{\textrm{LB}}^{k}\left(\cdot\right)$ depends on the
specific flow assignment $\{D_{m}^{k}\}_{k=1}^{K}$, which is determined
by the scheduling algorithm. To derive an approximation ratio against
the optimal schedule, we therefore require a global lower bound that
depends only on the original demand matrix $D_{m}$ and network parameters,
independent of any particular assignment and schedule.

Define $T_{\textrm{LB}}\left(\cdot\right)$ as the global CCT lower-bound
function of the traffic. Let $T_{\textrm{LB}}\left(D_{m}\right)$
denote the global lower bound of coflow $C_{m}$. For any demand matrix
$D_{m}\neq\mathbf{0}_{N\times N}$, we obtain 
\begin{equation}
T_{\textrm{LB}}\left(D_{m}\right)=\delta+\frac{\rho_{m}}{R}.\label{eq:global-lb}
\end{equation}

Let $T_{m}^{k}$ denote the completion time of the portion of coflow
$C_{m}$ assigned to core $k$, and let $\mathcal{K}_{m}=\{k:D_{m}^{k}\ne\mathbf{0}\}$
be the set of active cores for that coflow.

\begin{lemma}[\textit{Global Lower Bound}] In a $K$-core OCS network,
for any coflow $C_{m}$ with demand matrix $D_{m}$, the completion
time of any feasible schedule satisfies $T_{m}\ge T_{\textrm{LB}}\left(D_{m}\right)=\delta+\frac{\rho_{m}}{R}$.\label{single_global_lb}\end{lemma}
\begin{IEEEproof}
The per-core lower bound (Eq. (\ref{eq:per-core-lb})) can be relaxed
as 
\begin{equation}
\begin{aligned}T_{\textrm{LB}}^{k}\left(D_{m}^{k}\right) & \ge\max\left\{ \frac{\max_{i}\rho_{m,i}^{k}}{r^{k}}+\delta,\frac{\max_{j}\rho_{m,j}^{k}}{r^{k}}+\delta\right\} \\
 & =\frac{1}{r^{k}}\max\left\{ \max_{i}\rho_{m,i}^{k},\max_{j}\rho_{m,j}^{k}\right\} +\delta=\frac{\rho_{m}^{k}}{r^{k}}+\delta.
\end{aligned}
\label{eq:relaxed-core-lb}
\end{equation}

Since any feasible schedule on an active core $k$ satisfies $T_{m}^{k}\ge T_{\textrm{LB}}^{k}\left(D_{m}^{k}\right)$
and $T_{m}=\max_{k\in\mathcal{K}_{m}}T_{m}^{k}$, we obtain 
\begin{equation}
\begin{aligned}T_{m}\ge\max_{k\in\mathcal{K}_{m}}T_{\textrm{LB}}^{k}\left(D_{m}^{k}\right)\ge\delta+\max_{k\in\mathcal{K}_{m}}\frac{\rho_{m}^{k}}{r^{k}}.\end{aligned}
\label{eq:local-cct-lb}
\end{equation}

Since the maximum is no smaller than the weighted average, we have
\begin{equation}
\max_{k\in\mathcal{K}_{m}}\frac{\rho_{m}^{k}}{r^{k}}\ge\frac{\sum_{k\in\mathcal{K}_{m}}r^{k}\cdot\frac{\rho_{m}^{k}}{r^{k}}}{\sum_{k\in\mathcal{K}_{m}}r^{k}}=\frac{\sum_{k=1}^{K}\rho_{m}^{k}}{\sum_{k\in\mathcal{K}_{m}}r^{k}}\ge\frac{\sum_{k=1}^{K}\rho_{m}^{k}}{R}.
\end{equation}

Let $p^{*}$ be a port attaining the maximum load in $D_{m}$, so
that $\rho_{m}=\rho_{m,p^{*}}=\sum_{k=1}^{K}\rho_{m,p^{*}}^{k}$.
Since $\rho_{m}^{k}=\max_{p\in\mathcal{P}}\rho_{m,p}^{k}\geq\rho_{m,p^{*}}^{k}$,
it follows that 
\begin{equation}
\sum_{k=1}^{K}\rho_{m}^{k}\geq\sum_{k=1}^{K}\rho_{m,p^{*}}^{k}=\rho_{m}.
\end{equation}

Combining the above inequalities yields 
\begin{equation}
T_{m}\ge\delta+\frac{\rho_{m}}{R}=T_{\textrm{LB}}\left(D_{m}\right).\label{eq:global-cct-lb}
\end{equation}

This completes the proof.
\end{IEEEproof}

\subsection{Approximation Algorithm}

\subsubsection{Algorithm Framework}

Algorithm \ref{alg:alg1} consists of three components: (i) global
coflow ordering, (ii) cross-core flow assignment, and (iii) per-core
circuit scheduling.

\textbf{Global Coflow Ordering} We compute a global permutation $\pi$
over all coflows and use it as a common priority order across all
cores. Each coflow $C_{m}$ is assigned a priority score $w_{m}/T_{\textrm{LB}}\left(D_{m}\right)$,
where $T_{\textrm{LB}}\left(D_{m}\right)$ captures a fundamental
lower bound on the minimum completion time of $C_{m}$ in the multi-core
network. Coflows are then ordered in non-increasing order of this
score. This rule favors coflows with high weights and low inherent
service requirements, approximating the weighted shortest-processing-time
(WSPT) principle.

\textbf{Cross-Core Flow Assignment} Coflows are processed sequentially
according to $\pi$. Each flow is assigned entirely to a single core,
and flow splitting is prohibited to avoid packet reordering, buffering
overhead, and additional control-plane complexity in practical multi-core
OCS deployments \cite{huang2020weaver,chen2023scheduling}. Restricting
assignment to the flow granularity preserves analytical tractability
while maintaining practical implementability. Thus, under Algorithm
\ref{alg:alg1}, each original flow has exactly one nonzero assigned
component $f_{m}^{k}(i,j)$, which equals the entire original flow.

For the first $m$ coflows under $\pi$, let $\rho_{1:m,p}^{k}=\sum_{\ell=1}^{m}\rho_{\pi(\ell),p}^{k}$
be the cumulative traffic incident to port $p$ on core $k$, and
let $\tau_{1:m,p}^{k}=\sum_{\ell=1}^{m}\tau_{\pi(\ell),p}^{k}$ be
the cumulative number of assigned nonzero flows incident to port $p$
on core $k$. The workload state is $\mathcal{X}_{1:m}^{k}=\left(\{\rho_{1:m,p}^{k}\}_{p\in\mathcal{P}},\{\tau_{1:m,p}^{k}\}_{p\in\mathcal{P}}\right)$
and its bottleneck port workload is 
\begin{equation}
\Phi^{k}(\mathcal{X}_{1:m}^{k})=\max_{p\in\mathcal{P}}\left\{ \frac{\rho_{1:m,p}^{k}}{r^{k}}+\tau_{1:m,p}^{k}\delta\right\} .\label{eq:workload-measure}
\end{equation}
Each flow $f_{\pi(m)}(i,j)$ is assigned entirely to the core $k^{*}=\arg\min_{k\in\mathcal{K}}\Phi^{k}(\mathcal{X}_{1:m}^{k}\oplus f_{\pi(m)}(i,j))$,
where $\mathcal{X}_{1:m}^{k}\oplus f_{\pi(m)}(i,j)$ denotes the tentative
state obtained by assigning the flow to core $k$.

\textbf{Per-Core Circuit Scheduling} After assignment, each core applies
an event-driven greedy scheduler to all of its assigned flows. A scheduling
event occurs whenever one or more active flows complete and release
their incident ports. At each scheduling event, the scheduler scans
the unscheduled flows according to $\pi$ and constructs a maximal
set of port-disjoint feasible flows. The scheduler tracks the availability
time of each ingress and egress port. Each selected flow first incurs
the reconfiguration delay and then transmits non-preemptively to completion.
This rule has three properties:
\begin{itemize}
\item \textit{Port-Exclusive:} Each ingress and egress port participates
in at most one active circuit at any time, satisfying the one-to-one
port matching constraint of OCS.
\item \textit{Non-Preemptive:} Once a flow starts transmission, it proceeds
to completion without interruption, avoiding additional reconfiguration
overhead.
\item \textit{Maximal-Selection:} At each scheduling event, no additional
feasible flow can be added to the selected set without creating a
port conflict.
\end{itemize}

\subsubsection{Algorithm Design}

The algorithm operates in three phases. First, it determines the global
coflow priority order (Lines 1--4). For each coflow $C_{m}$, the
algorithm computes the score $\eta_{m}=\frac{w_{m}}{T_{\mathrm{LB}}(D_{m})}$,
where $T_{\textrm{LB}}\left(D_{m}\right)=\delta+\rho_{m}/R$ is a
lower bound on the completion time of $C_{m}$ (Line 2). The coflows
are sorted in non-increasing order of $\eta_{m}$ to obtain the global
priority order $\pi$ (Line 4).

\begin{algorithm}[!t]
\caption{Coflow Scheduling in Multi-Core OCS Networks}
\label{alg:alg1} \begin{algorithmic}[1] \Require Demand matrices
$\{D_{m}\}_{m=1}^{M}$, weights $\{w_{m}\}_{m=1}^{M}$, core rates
$\{r^{k}\}_{k=1}^{K}$, and reconfiguration delay $\delta$ \Ensure
Order $\pi$, assignments $\{D_{m}^{k}\}$, and schedules $\{S_{m}^{k}\}$
\Linecomment{COFLOW ORDERING} \For{$m=1$ to $M$} \State $\eta_{m}\gets w_{m}/T_{\textrm{LB}}(D_{m})$
\EndFor \State $\pi\gets$ coflows sorted in non-increasing order
of $\eta_{m}$ \Linecomment{FLOW ASSIGNMENT} \State Initialize
$\mathcal{X}_{1:0}^{k}\gets\mathbf{0}$ and $D_{m}^{k}\gets\mathbf{0}$
for all $m,k$ \For{$m=1$ to $M$} \State $\mathcal{X}_{1:m}^{k}\gets\mathcal{X}_{1:m-1}^{k}$
for all $k$ \State $\mathcal{F}_{\pi(m)}\gets\{f_{\pi(m)}(i,j):d_{\pi(m)}(i,j)>0\}$
\State Sort $\mathcal{F}_{\pi(m)}$ in non-increasing order of flow
size \For{each $f_{\pi(m)}(i,j)\in\mathcal{F}_{\pi(m)}$} \State
$k^{*}\gets\arg\min_{k}\Phi^{k}(\mathcal{X}_{1:m}^{k}\oplus f_{\pi(m)}(i,j))$
\State Assign the entire flow to core $k^{*}$ and update $D_{\pi(m)}^{k^{*}}$
\State $\mathcal{X}_{1:m}^{k^{*}}\gets\mathcal{X}_{1:m}^{k^{*}}\oplus f_{\pi(m)}(i,j)$
\EndFor \EndFor \Linecomment{CIRCUIT SCHEDULING} \For{$k=1$
to $K$} \State Initialize $S_{m}^{k}\gets\emptyset$ for all $m$
and set $t\gets0$ \State Set $b_{p}^{k}\gets0$ for every ingress/egress
port $p$ \State $\mathcal{F}^{k}\gets\bigcup_{m}\{f_{m}^{k}(i,j):d_{m}^{k}(i,j)>0\}$
\State Sort $\mathcal{F}^{k}$ according to $\pi$ \While{$\mathcal{F}^{k}\neq\emptyset$}
\State $\mathcal{U}\gets\emptyset$ \For{each flow $\ensuremath{f_{m}^{k}(i,j)\in\mathcal{F}^{k}}$}
\If{$b_{i}^{k}\le t$, $b_{j}^{k}\le t$, ports $i,j$ are unused
by $\mathcal{U}$} \State Add the flow to $\mathcal{U}$ \EndIf
\EndFor \If{$\mathcal{U}=\emptyset$} \State Advance $t$ to
the next flow completion time \Else \For{each $f_{m}^{k}(i,j)\in\mathcal{U}$}
\State $t_{m}^{k}(i,j)\gets t$ \State $T_{m}^{k}(i,j)\gets t+\delta+d_{m}^{k}(i,j)/r^{k}$
\State $b_{i}^{k}\gets T_{m}^{k}(i,j)$ and $b_{j}^{k}\gets T_{m}^{k}(i,j)$
\State Add $(i,j,t)$ to $S_{m}^{k}$ and remove the flow from $\mathcal{F}^{k}$
\EndFor \EndIf \EndWhile \EndFor \end{algorithmic}
\end{algorithm}

Next, the algorithm performs cross-core flow assignment (Lines 5-15).
For each core $k$, the state $\mathcal{X}_{1:m}^{k}$ records the
cumulative traffic load and cumulative number of assigned nonzero
flows incident to each port for the first $m$ coflows under $\pi$.
All workload states and assignment matrices are initialized to zero
(Line 5). Before assigning $C_{\pi(m)}$, each core inherits the workload
state of the preceding prefix by setting $\mathcal{X}_{1:m}^{k}\gets\mathcal{X}_{1:m-1}^{k}$
(Line 7). The nonzero flows of $C_{\pi(m)}$ are collected and sorted
in non-increasing order of size (Lines 8--9). For each flow $f_{\pi(m)}(i,j)$,
the algorithm evaluates its tentative assignment to each core. The
operation $\mathcal{X}_{1:m}^{k}\oplus f_{\pi(m)}(i,j)$ adds $d_{\pi(m)}(i,j)$
to the traffic loads of ports $i$ and $j$, and increases the nonzero
assigned flow count of each incident port by one. The flow is assigned
entirely to the core $k^{*}=\arg\min_{k}\Phi^{k}\left(\mathcal{X}_{1:m}^{k}\oplus f_{\pi(m)}(i,j)\right)$
(Line 11). The corresponding assignment matrix and workload state
are then updated (Lines 12--13).

After assignment, each core schedules its assigned flows independently
(Lines 16--39). For core $k$, the algorithm initializes all schedules
and sets the current time to $t=0$ (Line 17). For each port $p$,
$b_{p}^{k}$ denotes its availability time: the earliest time at which
it can participate in a new circuit. Initially, $b_{p}^{k}=0$ for
each port (Line 18). The set $\mathcal{F}^{k}$ contains all unscheduled
flows assigned to core $k$ (Line 19). These flows are ordered according
to $\pi$ (Line 20). At each event time $t$, the algorithm scans
the remaining flows in this order and constructs a set $\mathcal{U}$
(Lines 22--27). A flow is added to $\mathcal{U}$ if both incident
ports are available at time $t$ and neither port is already used
by another flow in $\mathcal{U}$ (Line 25). At the end of the scan,
$\mathcal{U}$ is maximal among the flows feasible at time $t$. If
$\mathcal{U}$ is empty, the algorithm advances $t$ to the next flow
completion time (Lines 28--29). Otherwise, reconfiguration for each
flow $f_{m}^{k}(i,j)\in\mathcal{U}$ begins at time $t$ and the flow
completes at $t+\delta+\frac{d_{m}^{k}(i,j)}{r^{k}}$ (Line 33). The
availability times of ports $i$ and $j$ are updated to $T_{m}^{k}(i,j)$
(Line 34). The scheduled flow is then recorded and removed from $\mathcal{F}^{k}$
(Line 35). Once transmission begins, the flow proceeds non-preemptively
to completion. The order $\pi$ determines the priority among feasible
flows but does not impose a precedence constraint. Thus, a lower-priority
flow may start when a higher-priority flow is blocked by unavailable
ports.

\subsection{Analysis of Performance Guarantees\label{sec:Theoretical-Analysis_Multiple}}

Let $\pi$ denote the global coflow order produced by Algorithm \ref{alg:alg1}.
For any $m\in\{1,\ldots,M\}$ and $p\in\mathcal{P}$, define the cumulative
traffic load $\rho_{1:m,p}=\sum_{\ell=1}^{m}\rho_{\pi(\ell),p}$ and
the cumulative number of nonzero flows $\tau_{1:m,p}=\sum_{\ell=1}^{m}\tau_{\pi(\ell),p}$
incident to port $p$ for the first $m$ coflows under $\pi$.

\subsubsection{Assignment-Phase Prefix Bound}

Define $\rho_{1:m}=\max_{p\in\mathcal{P}}\rho_{1:m,p}$ and $\rho_{1:m}^{k}=\max_{p\in\mathcal{P}}\rho_{1:m,p}^{k}$
as the maximum cumulative port loads globally and on core $k$, respectively.
Similarly, define $\tau_{1:m}=\max_{p\in\mathcal{P}}\tau_{1:m,p}$
and $\tau_{1:m}^{k}=\max_{p\in\mathcal{P}}\tau_{1:m,p}^{k}$ as the
maximum cumulative flow counts. Finally, let $r_{\max}=\max_{k\in\mathcal{K}}r^{k}$.

\begin{lemma}[Assignment-Phase Prefix Bound] \label{multiple-assignment-stage}
For every $m=1,\ldots,M$, the states generated by Algorithm \ref{alg:alg1}
satisfy $\max_{k}\Phi^{k}(\mathcal{X}_{1:m}^{k})\le\frac{\rho_{1:m}}{r_{\max}}+\tau_{1:m}\delta$.
\end{lemma}
\begin{IEEEproof}
Fix any nonempty core $k_{1}$ after the first $m$ coflows have been
assigned. Let $\bar{f}$ be the last flow among the first $m$ coflows
that is assigned to $k_{1}$, and let $\bar{\mathcal{X}}^{k}$ denote
the state of core $k$ immediately before that assignment. Hence,
\begin{equation}
\mathcal{X}_{1:m}^{k_{1}}=\bar{\mathcal{X}}^{k_{1}}\oplus\bar{f}.
\end{equation}

By the greedy assignment rule, for each core $k_{2}\in\mathcal{K}$,
\begin{equation}
\Phi^{k_{1}}(\bar{\mathcal{X}}^{k_{1}}\oplus\bar{f})\le\Phi^{k_{2}}(\bar{\mathcal{X}}^{k_{2}}\oplus\bar{f}).
\end{equation}

The tentative state on $k_{2}$ contains only flows from the first
$m$ coflows. Hence, each port has traffic load at most $\rho_{1:m}$
and the number of assigned nonzero flows incident to each port is
at most $\tau_{1:m}$. Therefore, 
\begin{equation}
\Phi^{k_{2}}(\bar{\mathcal{X}}^{k_{2}}\oplus\bar{f})\le\frac{\rho_{1:m}}{r^{k_{2}}}+\tau_{1:m}\delta.
\end{equation}

Combining the above inequalities and taking the minimum over $k_{2}$,
we obtain
\begin{equation}
\Phi^{k_{1}}\left(\mathcal{X}_{1:m}^{k_{1}}\right)\le\min_{k_{2}\in\mathcal{K}}\left\{ \frac{\rho_{1:m}}{r^{k_{2}}}+\tau_{1:m}\delta\right\} =\frac{\rho_{1:m}}{r_{\max}}+\tau_{1:m}\delta.
\end{equation}

Since $k_{1}$ is an arbitrary nonempty core, taking the maximum over
all nonempty cores gives
\begin{equation}
\max_{k}\Phi^{k}(\mathcal{X}_{1:m}^{k})\le\frac{\rho_{1:m}}{r_{\max}}+\tau_{1:m}\delta.
\end{equation}

This completes the proof.
\end{IEEEproof}

\subsubsection{Scheduling-Phase Makespan Bound}

Let $\ensuremath{T_{\max}^{k}}=\max\{T_{m}^{k}(i,j):d_{m}^{k}(i,j)>0\}$
denote the time at which core $k$ completes all its assigned flows.
Define the overall makespan as $\ensuremath{T_{\max}}=\max_{k}\ensuremath{T_{\max}^{k}}$.

\begin{lemma}[Scheduling-Phase Makespan Bound] \label{multiple-scheduling-stage}
The circuit scheduling phase of Algorithm \ref{alg:alg1} satisfies
$T_{\max}\le2\max_{k}\Phi^{k}(\mathcal{X}_{1:M}^{k})$.\end{lemma}
\begin{IEEEproof}
Fix a nonempty core $k$, let $f_{m}^{k}(i^{\star},j^{\star})$ be
the last flow completed on that core, and its completion time is therefore
$\ensuremath{T_{\max}^{k}}$. Let $\ensuremath{d^{\star}=d_{m}^{k}(i^{\star},j^{\star})}$
be its size, and let $t^{\star}$ be the time at which its reconfiguration
begins. Hence, we can get 
\begin{equation}
\ensuremath{T_{\max}^{k}}=t^{\star}+\delta+\frac{d^{\star}}{r^{k}}.\label{eq:last-flow-completion}
\end{equation}

Let $0=t_{0}<t_{1}<\cdots<t_{q}=t^{\star}$ denote the consecutive
scheduling event times up to $t^{\star}$, where each event is triggered
by the completion of at least one scheduled flow. Consider any interval
$[t_{\ell},t_{\ell+1})$, no scheduled flow completes within this
interval. Therefore, if either $i^{\star}$ or $j^{\star}$ is unavailable
at $t_{\ell}$, that port remains unavailable throughout this interval.
Otherwise, both ports are available at $t_{\ell}$. Because $t_{\ell}<t^{\star}$,
flow $f_{m}^{k}(i^{\star},j^{\star})$ is not selected at this event.
Since Algorithm \ref{alg:alg1} constructs a maximal set of feasible
port-disjoint flows, some selected flow must use $i^{\star}$ or $j^{\star}$;
otherwise, $f_{m}^{k}(i^{\star},j^{\star})$ could also be added.
Such a selected flow begins reconfiguration at $t_{\ell}$ and occupies
the shared port continuously until its completion, which occurs no
earlier than $t_{\ell+1}$. Hence, at least one of $i^{\star}$ and
$j^{\star}$ is unavailable throughout every interval $[t_{\ell},t_{\ell+1})$.

Let $B_{i^{\star}}(t^{\star})$ and $B_{j^{\star}}(t^{\star})$ denote
the total durations for which ports $i^{\star}$ and $j^{\star}$
are unavailable during $[0,t^{\star})$, respectively. These durations
include both reconfiguration and transmission. Since $f_{m}^{k}(i^{\star},j^{\star})$
has not started before $t^{\star}$, only the other flows incident
to $i^{\star}$ and $j^{\star}$ can contribute to these busy times.
Hence,
\begin{align}
B_{i^{\star}}(t^{\star}) & \le\frac{\rho_{1:M,i^{\star}}^{k}-d^{\star}}{r^{k}}+(\tau_{1:M,i^{\star}}^{k}-1)\delta,\\
B_{j^{\star}}(t^{\star}) & \le\frac{\rho_{1:M,j^{\star}}^{k}-d^{\star}}{r^{k}}+(\tau_{1:M,j^{\star}}^{k}-1)\delta.
\end{align}

Hence, at least one of the two ports is unavailable at every time
in $[0,t^{\star})$, and therefore $t^{\star}\le B_{i^{\star}}(t^{\star})+B_{j^{\star}}(t^{\star})$.
Combining the above inequalities gives
\begin{equation}
\begin{aligned}\ensuremath{T_{\max}^{k}} & =t^{\star}+\delta+\frac{d^{\star}}{r^{k}}\\[6pt]
 & \le\frac{\rho_{1:M,i^{\star}}^{k}+\rho_{1:M,j^{\star}}^{k}-d^{\star}}{r^{k}}+\left(\tau_{1:M,i^{\star}}^{k}+\tau_{1:M,j^{\star}}^{k}-1\right)\delta\\[6pt]
 & \le\left(\frac{\rho_{1:M,i^{\star}}^{k}}{r^{k}}+\tau_{1:M,i^{\star}}^{k}\delta\right)+\left(\frac{\rho_{1:M,j^{\star}}^{k}}{r^{k}}+\tau_{1:M,j^{\star}}^{k}\delta\right)\\
 & \le2\Phi^{k}\left(\mathcal{X}_{1:M}^{k}\right).
\end{aligned}
\end{equation}

Taking the maximum over all cores yields 
\begin{equation}
\ensuremath{T_{\max}=\max_{k}\ensuremath{T_{\max}^{k}}\le2\max_{k}\Phi^{k}\left(\mathcal{X}_{1:M}^{k}\right).}
\end{equation}

This completes the proof.
\end{IEEEproof}

\subsubsection{Deterministic Approximation Ratio (Algorithm \ref{alg:alg1})}

Let $T_{m}$ denote the completion time of coflow $C_{m}$ produced
by Algorithm \ref{alg:alg1}, and let $T_{m}^{*}$ denote its completion
time in an optimal schedule for minimizing total weighted CCT. Define
$\ensuremath{w_{\max}=\max_{1\le m\le M}w_{m}}$ and $w_{\min}=\min_{1\le m\le M}w_{m}$,
$\ensuremath{\tau_{\max}=\max_{1\le m\le M}\tau_{m}}$, and $\psi=\max\{K,\tau_{\max}\}$.
For each coflow $C_{m}$, define $a_{m}=\frac{\rho_{m}}{r_{\max}}+\tau_{m}\delta$.
Since each row or column of $D_{m}$ contains at most $N$ nonzero
flows, we have $\tau_{\max}\le N$.

\begin{lemma}[Auxiliary CCT Lower Bound] \label{cct-lower-bound}For
every coflow $C_{m}$, we have $T_{m}^{*}\ge T_{\textrm{LB}}(D_{m})\ge\frac{a_{m}}{\psi}$.\end{lemma}
\begin{IEEEproof}
By Lemma \ref{single_global_lb}, $\ensuremath{T_{m}^{*}\ge T_{\mathrm{LB}}(D_{m})=\delta+\frac{\rho_{m}}{R}}$.
Since $R=\ensuremath{\sum_{k=1}^{K}r^{k}}\le Kr_{\max}\le\psi r_{\max}$,
we have 
\begin{equation}
\frac{\rho_{m}}{R}\ge\frac{\rho_{m}}{\psi r_{\max}}.
\end{equation}

Similarly, $\tau_{m}\le\tau_{\max}\le\psi$ implies $\delta\ge\tau_{m}\delta/\psi$.
Therefore, 
\begin{equation}
T_{\textrm{LB}}(D_{m})=\delta+\frac{\rho_{m}}{R}\ge\frac{1}{\psi}\left(\frac{\rho_{m}}{r_{\max}}+\tau_{m}\delta\right)=\frac{a_{m}}{\psi}.
\end{equation}

This completes the proof.
\end{IEEEproof}
\begin{thm}
\label{theorem1} Algorithm \ref{alg:alg1} achieves a $2M\frac{w_{\max}}{w_{\min}}\psi$-approximation
for minimizing the total weighted CCT in a multi-core OCS network,
i.e., $\sum_{m=1}^{M}w_{m}T_{m}\le2M\frac{w_{\max}}{w_{\min}}\psi\sum_{m=1}^{M}w_{m}T_{m}^{*}$,
where $\psi=\max\left\{ K,\tau_{\max}\right\} $, $\tau_{\max}\leq N$,
and $N$ is the number of ingress/egress ports per core.
\end{thm}
\begin{IEEEproof}
Every coflow completes no later than the makespan $T_{\max}$. Hence,
$T_{m}\le T_{\max}$ for each $m$. Applying Lemma \ref{multiple-assignment-stage}
with $m=M$, followed by Lemma \ref{multiple-scheduling-stage} gives
\begin{equation}
\begin{aligned}T_{\max} & \le2\left(\frac{\rho_{1:M}}{r_{\max}}+\tau_{1:M}\delta\right)\\
 & \le2\sum_{m=1}^{M}\left(\frac{\rho_{m}}{r_{\max}}+\tau_{m}\delta\right)=2\sum_{m=1}^{M}a_{m}.
\end{aligned}
\end{equation}
where we use $\rho_{1:M}\le\sum_{m=1}^{M}\rho_{m}$ and $\tau_{1:M}\le\sum_{m=1}^{M}\tau_{m}$.

Therefore, the total weighted CCT produced by Algorithm \ref{alg:alg1}
satisfies 
\begin{equation}
\sum_{m=1}^{M}w_{m}T_{m}\le T_{\max}\sum_{m=1}^{M}w_{m}\le2Mw_{\max}\sum_{m=1}^{M}a_{m}.
\end{equation}

On the other hand, Lemma \ref{cct-lower-bound} gives 
\begin{equation}
\sum_{m=1}^{M}w_{m}T_{m}^{*}\ge w_{\min}\sum_{m=1}^{M}T_{m}^{*}\ge\frac{w_{\min}}{\psi}\sum_{m=1}^{M}a_{m}.
\end{equation}

Consequently, 
\begin{equation}
\begin{aligned}\frac{\sum_{m=1}^{M}w_{m}T_{m}}{\sum_{m=1}^{M}w_{m}T_{m}^{*}} & \le\frac{2Mw_{\max}\sum_{m=1}^{M}a_{m}}{\frac{w_{\min}}{\psi}\sum_{m=1}^{M}a_{m}}=2M\frac{w_{\max}}{w_{\min}}\psi.
\end{aligned}
\end{equation}

This completes the proof.
\end{IEEEproof}
Based on Theorem \ref{theorem1}, we can derive the following corollaries.

\begin{Corollary} In the unweighted case (i.e., $w_{m}=1$ for all
$m$), Algorithm \ref{alg:alg1} is $2M\psi$-approximation for minimizing
the total CCT in a multi-core OCS network, where $\psi=\max\left\{ K,\tau_{\max}\right\} $.
\end{Corollary}

\begin{Corollary} In the single-coflow case (i.e., $M=1$), Algorithm
\ref{alg:alg1} is $2\psi$-approximation for minimizing the CCT in
a multi-core OCS network, where $\psi=\max\left\{ K,\tau_{\max}\right\} $.
\end{Corollary}

In addition to the worst-case approximation ratio $\ensuremath{2\psi M\frac{w_{\max}}{w_{\min}}}$,
we derive a refined deterministic approximation ratio for multi-core
OCS networks that incorporates both coflow weights and workloads
(Theorem \ref{theorem2}).

\subsubsection{Refined Approximation Ratio}

Recall that $a_{m}=\frac{\rho_{m}}{r_{\max}}+\tau_{m}\delta,\psi=\max\left\{ K,\tau_{\max}\right\} $,
and Lemma \ref{cct-lower-bound} gives $T_{m}^{*}\ge a_{m}/\psi$.
Let $W=\sum_{m=1}^{M}w_{m}$ and define $\Gamma_{w,a}=\frac{W\sum_{m=1}^{M}a_{m}}{\sum_{m=1}^{M}w_{m}a_{m}}.$
Since every coflow is nonzero and every weight is positive, the denominator
is strictly positive. The factor $\ensuremath{\Gamma_{w,a}}$ captures
the relationship between the coflow weights and workload quantities
$a_{m}$.
\begin{thm}
\label{theorem2} For arbitrary positive weights, Algorithm \ref{alg:alg1}
satisfies $\sum_{m=1}^{M}w_{m}T_{m}\le2\psi\Gamma_{w,a}\sum_{m=1}^{M}w_{m}T_{m}^{*}.$
Moreover, $\Gamma_{w,a}\le\frac{W}{w_{\min}}\le M\frac{w_{\max}}{w_{\min}}$.
Therefore, this bound is never weaker than the approximation ratio
in Theorem \ref{theorem1}.
\end{thm}
\begin{IEEEproof}
The proof of Theorem \ref{theorem1} establishes $T_{\max}\le2\sum_{m=1}^{M}a_{m}.$
Since $T_{m}\le T_{\max}$ for every coflow $C_{m}$, we have 
\begin{equation}
\sum_{m=1}^{M}w_{m}T_{m}\le T_{\max}\sum_{m=1}^{M}w_{m}=T_{\max}W\le2W\sum_{m=1}^{M}a_{m}.
\end{equation}

Lemma \ref{cct-lower-bound} gives $\sum_{m=1}^{M}w_{m}T_{m}^{*}\ge\frac{1}{\psi}\sum_{m=1}^{M}w_{m}a_{m}.$
Combining these inequalities yields 
\begin{equation}
\ensuremath{\begin{aligned}\frac{\sum_{m=1}^{M}w_{m}T_{m}}{\sum_{m=1}^{M}w_{m}T_{m}^{*}} & \le2\psi\frac{W\sum_{m=1}^{M}a_{m}}{\sum_{m=1}^{M}w_{m}a_{m}}\\
 & =2\psi\Gamma_{w,a}.
\end{aligned}
}
\end{equation}

Furthermore, 
\begin{equation}
\sum_{m=1}^{M}w_{m}a_{m}\ge w_{\min}\sum_{m=1}^{M}a_{m}.
\end{equation}

Hence, 
\begin{equation}
\Gamma_{w,a}=\frac{W\sum_{m=1}^{M}a_{m}}{\sum_{m=1}^{M}w_{m}a_{m}}\le\frac{W}{w_{\min}}.
\end{equation}

Since $W\le Mw_{\max}$, 
\begin{equation}
\Gamma_{w,a}\le\frac{W}{w_{\min}}\le M\frac{w_{\max}}{w_{\min}}.
\end{equation}

This completes the proof.
\end{IEEEproof}

\subsubsection{Deterministic Approximation Ratio (EPS variant of Algorithm \ref{alg:alg1})}

The framework of Algorithm \ref{alg:alg1} can also be applied to
a multi-core EPS network by setting $\delta=0$ and replacing the
OCS lower bounds $T_{\textrm{LB}}^{k}\left(\cdot\right)$ and $T_{\textrm{LB}}\left(\cdot\right)$
with their EPS counterparts.

Consider an $H$-core EPS network, where each core $h\in\left\{ 1,...,H\right\} $
provides per-port transmission rate $\widetilde{r}^{h}$. Define $\ensuremath{\widetilde{R}=\sum_{h=1}^{H}\widetilde{r}^{h}}$
and $\widetilde{r}_{\max}=\max_{h}\widetilde{r}^{h}$. Let $\widetilde{T}_{m}$
denote the completion time of coflow $C_{m}$ produced by the EPS
variant of Algorithm \ref{alg:alg1}. Let $\widetilde{T}_{m}^{*}$
denote its completion time in an optimal schedule for minimizing total
weighted CCT. For any demand matrix $D_{m}$, the per-core EPS lower
bound is $\widetilde{T}_{\textrm{LB}}^{h}\left(D_{m}^{h}\right)=\frac{\rho_{m}^{h}}{\widetilde{r}^{h}}$,
and the global EPS lower bound is $\widetilde{T}_{m}^{*}\geq\widetilde{T}_{\textrm{LB}}\left(D_{m}\right)=\frac{\rho_{m}}{\widetilde{R}}$,
where $\rho_{m}^{h}$ and $\rho_{m}$ denote the maximum port loads
of $D_{m}^{h}$ and $D_{m}$, respectively \cite{huang2020weaver}.
\begin{thm}
The EPS variant of Algorithm \ref{alg:alg1} achieves a $2MH\frac{w_{\max}}{w_{\min}}$-approximation
for minimizing the total weighted CCT in a multi-core EPS network,
i.e., $\sum_{m=1}^{M}w_{m}\widetilde{T}_{m}\le2MH\frac{w_{\max}}{w_{\min}}\sum_{m=1}^{M}w_{m}\widetilde{T}_{m}^{*}$.\label{theorem4}
\end{thm}
\begin{IEEEproof}
The EPS variant retains the same assignment and scheduling framework
as Algorithm \ref{alg:alg1} with $\delta=0$. Applying the assignment
bound to the complete workload and then applying the makespan bound
gives 
\begin{equation}
\ensuremath{\widetilde{T}_{\max}}\le\frac{2\rho_{1:M}}{\widetilde{r}_{\max}}\le\frac{2}{\widetilde{r}_{\max}}\sum_{m=1}^{M}\rho_{m},
\end{equation}
where $\widetilde{T}_{\max}$ denotes the makespan produced by the
EPS variant. Since each coflow completes no later than $\widetilde{T}_{\max}$,
\begin{equation}
\sum_{m=1}^{M}w_{m}\widetilde{T}_{m}\le\widetilde{T}_{\max}\sum_{m=1}^{M}w_{m}\le\frac{2Mw_{\max}}{\widetilde{r}_{\max}}\sum_{m=1}^{M}\rho_{m}.
\end{equation}

The global EPS lower bound gives 
\begin{equation}
\sum_{m=1}^{M}w_{m}\widetilde{T}_{m}^{*}\ge\sum_{m=1}^{M}w_{m}\frac{\rho_{m}}{\widetilde{R}}\ge\frac{w_{\min}}{\widetilde{R}}\sum_{m=1}^{M}\rho_{m}.
\end{equation}

Therefore, 
\begin{equation}
\begin{aligned}\frac{\sum_{m=1}^{M}w_{m}\widetilde{T}_{m}}{\sum_{m=1}^{M}w_{m}\widetilde{T}_{m}^{*}} & \le2M\frac{w_{\max}}{w_{\min}}\cdot\frac{\widetilde{R}}{\widetilde{r}_{\max}}\le2MH\frac{w_{\max}}{w_{\min}}.\\[4pt]\end{aligned}
\end{equation}

This completes the proof.
\end{IEEEproof}
Based on Theorem \ref{theorem4}, we can further derive the following
corollaries.

\begin{Corollary} In the unweighted case (i.e., $w_{m}=1$ for all
$m$), the EPS variant of Algorithm \ref{alg:alg1} is a $2MH$-approximation
for minimizing the total unweighted CCT in a multi-core EPS network.
\end{Corollary}

\begin{Corollary} In the single-coflow case (i.e., $M=1$), the EPS
variant of Algorithm \ref{alg:alg1} is a $2H$-approximation for
minimizing the CCT in a multi-core EPS network. \end{Corollary}

\section{Experimental Evaluations}
\label{sec:Experimental-Evaluations}
We assess Algorithm~\ref{alg:alg1} on coflow instances constructed
from the Facebook trace~\cite{facebook}, measuring both total
weighted CCT and tail completion times.

\subsection{Experimental Setup}
\textbf{Default Settings:}
Unless otherwise specified, we use $N=16$ ingress and egress
ports and $K=3$ OCS cores with rates $[10,20,30]\,\mathrm{GB/s}$,
giving an aggregate port rate of
$R=\sum_{k=1}^{K}r^k=60\,\mathrm{GB/s}$.
The reconfiguration delay is $\delta=8\,\mathrm{ms}$.
Each instance contains $M=100$ randomly sampled coflows,
whose weights are drawn independently from
$\operatorname{Unif}\{1,\ldots,10\}$.
Data volumes are measured in $\mathrm{MB}$, core rates in
$\mathrm{GB/s}$, and all resulting times in $\mathrm{ms}$.

\textbf{Trace Preparation:}
We use the Facebook trace~\cite{facebook}, which contains
526 coflows collected from 3000 machines in 150 racks
and has been used in previous studies
\cite{sunflow,regularization,wang2024scheduling,wang2023efficient,wang2023online,wang2025optimal,huang2020weaver}.
The trace records mapper and reducer rack locations,
each reducer's aggregate shuffle volume, and arrival
timestamps, but does not specify individual flow sizes.

To obtain flow demands, we divide each reducer's volume
pseudo-uniformly among its associated mapper racks,
with a random perturbation of $\pm10\%$ to avoid equal
flow sizes. We then randomly select $N$ rack-level ports
to form the $N\times N$ demand matrices.
All coflows are available at time zero, so the original
arrival timestamps are ignored.

\textbf{Performance Metrics:}
We compare total weighted CCT using Algorithm~\ref{alg:alg1},
denoted by $\textsc{Ours}$, as the reference.
For each scheme $\mathcal{A}$, the normalized cost is
\begin{equation}
	\mathrm{NormW}(\mathcal{A})
	\triangleq
	\frac{\sum_{m=1}^{M}w_mT_m(\mathcal{A})}
	{\sum_{m=1}^{M}w_mT_m(\textsc{Ours})}.
	\label{eq:normW}
\end{equation}
Thus, $\mathrm{NormW}(\textsc{Ours})=1$, and values above
one indicate worse performance.
We also evaluate tail performance using the p95 and p99 CCTs,
each divided by its corresponding value under $\textsc{Ours}$.

\subsection{Baseline Solutions}

We evaluate the contribution of flow assignment and circuit
scheduling through four variants of Algorithm~\ref{alg:alg1}
($\textsc{Ours}$). All variants retain its global coflow
priority order.

\begin{itemize}
	\item \textbf{RHO-ASSIGN:}
	Assign each flow to the core that minimizes the resulting
	$\rho_{1:m}^{k}/r^{k}$, considering transmission load
	without reconfiguration overhead. The per-core circuit
	scheduler is unchanged.
	
	\item \textbf{RAND-ASSIGN:}
	Select a core randomly for each flow, with selection
	probability proportional to its rate $r^{k}$.
	The per-core circuit scheduler is unchanged.
	
	\item \textbf{SUNFLOW-CORE:}
	Keep the flow assignment of $\textsc{Ours}$, but replace
	its per-core scheduler with Sunflow's multi-coflow
	heuristic~\cite{sunflow} for not-all-stop reconfiguration.
	
	\item \textbf{RAND-SUNFLOW:}
	Combine the rate-proportional random assignment of
	\textsc{RAND-ASSIGN} with the per-core Sunflow heuristic
	used in \textsc{SUNFLOW-CORE}.
\end{itemize}

\subsection{Experimental Results}
We first compare the baseline variants under the default
configuration, then examine performance across core counts
$K$, core rates $r^k$, port counts $N$, coflow counts $M$,
and reconfiguration delays $\delta$.

\subsubsection{Ablation under the Default Setting}
Fig.~\ref{fig:weighted_tail_cct} compares total weighted CCT
and the p95 and p99 CCTs, each normalized to $\textsc{Ours}$.
RHO-ASSIGN and RAND-ASSIGN yield total weighted CCTs of
$1.64\times$ and $1.31\times$ that of $\textsc{Ours}$,
respectively. RHO-ASSIGN also incurs approximately
$1.67\times$ the tail CCT, while both tail metrics
of RAND-ASSIGN remain above those of $\textsc{Ours}$.
These results suggest that accounting for reconfiguration
overhead improves flow assignment in this setting.

\begin{figure}[h]
	\centering\includegraphics[width=9cm,totalheight=8cm,keepaspectratio,height=5.8cm]{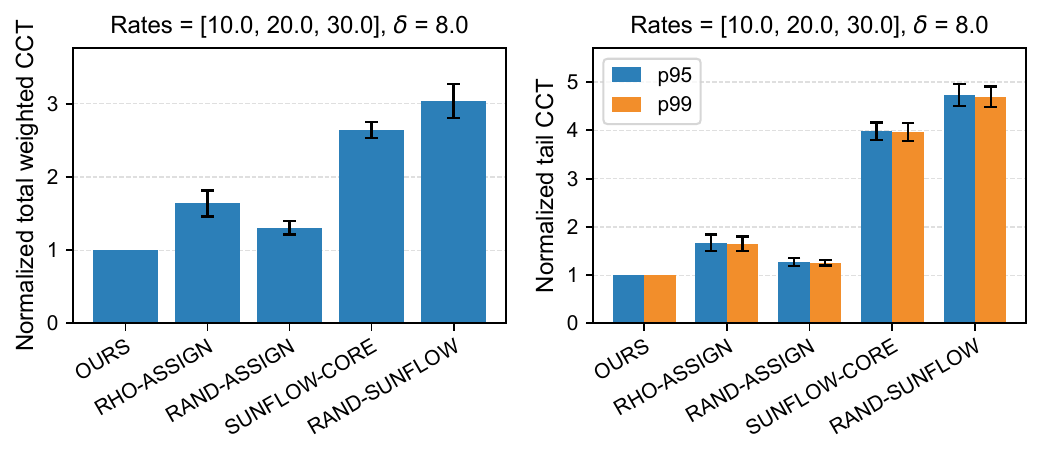}
	\caption{Normalized total weighted CCT and tail CCT under the default setting.}
	\label{fig:weighted_tail_cct}
\end{figure}

Replacing the per-core circuit scheduler with Sunflow leads
to a much larger degradation: SUNFLOW-CORE reaches a
normalized total weighted CCT of $2.64$, while both p95
and p99 CCTs approach $4\times$ those of $\textsc{Ours}$.
RAND-SUNFLOW performs worst among the compared schemes,
with $3.03\times$ the total weighted CCT and about
$4.7\times$ the p95 and p99 CCTs of $\textsc{Ours}$.

\subsubsection{Impact of Reconfiguration Delay ($\delta$-Sensitivity)}

We fix $N=16$ and $M=100$ and vary
$\delta\in\{2,4,6,8,10,12\}\,\mathrm{ms}$.
For each $K\in\{3,4,5\}$, we compare heterogeneous
and homogeneous core rates, as shown in
Figs.~\ref{fig:delta_k3}--\ref{fig:delta_k5}.

\begin{itemize}
	\item $K=3$ (Fig.~\ref{fig:delta_k3}).
	The relative performance gaps remain broadly stable
	as $\delta$ increases.
	Under heterogeneous rates, RHO-ASSIGN and RAND-ASSIGN
	incur approximately $1.4\times$ and $1.3\times$
	the total weighted CCT of $\textsc{Ours}$, respectively,
	while SUNFLOW-CORE and RAND-SUNFLOW have larger gaps.
	Under homogeneous rates, the gaps are generally smaller,
	with RAND-ASSIGN performing close to $\textsc{Ours}$.
\end{itemize}

\begin{figure}[h]
	\centering
	\includegraphics[width=9cm,totalheight=8cm,keepaspectratio,height=5.8cm]{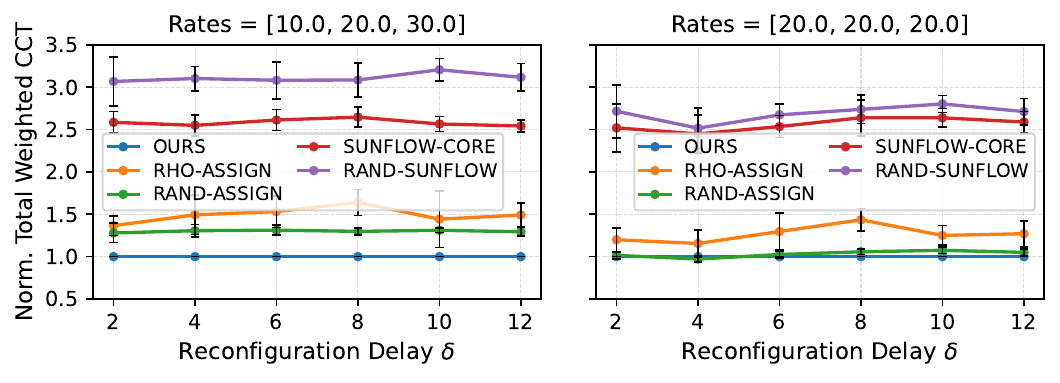}
	\caption{Normalized total weighted CCT versus $\delta$ for $K=3$.}
	\label{fig:delta_k3}
\end{figure}

\begin{itemize}
	\item $K=4$ (Fig.~\ref{fig:delta_k4}).
	Under heterogeneous rates, RHO-ASSIGN yields approximately
	$1.45$--$1.75$ times the total weighted CCT of
	$\textsc{Ours}$, while RAND-ASSIGN yields
	$1.34$--$1.43$ times.
	Under homogeneous rates, these ranges become
	$1.22$--$1.46$ for RHO-ASSIGN and approximately
	$1.03$--$1.06$ for RAND-ASSIGN.
	In this homogeneous setting, SUNFLOW-CORE and
	RAND-SUNFLOW remain further from $\textsc{Ours}$,
	at approximately $2.78$--$2.88$ and $2.90$--$3.07$
	times its total weighted CCT, respectively.
\end{itemize}

\begin{figure}[h]
	\centering
	\includegraphics[width=9cm,totalheight=8cm,keepaspectratio,height=5.8cm]{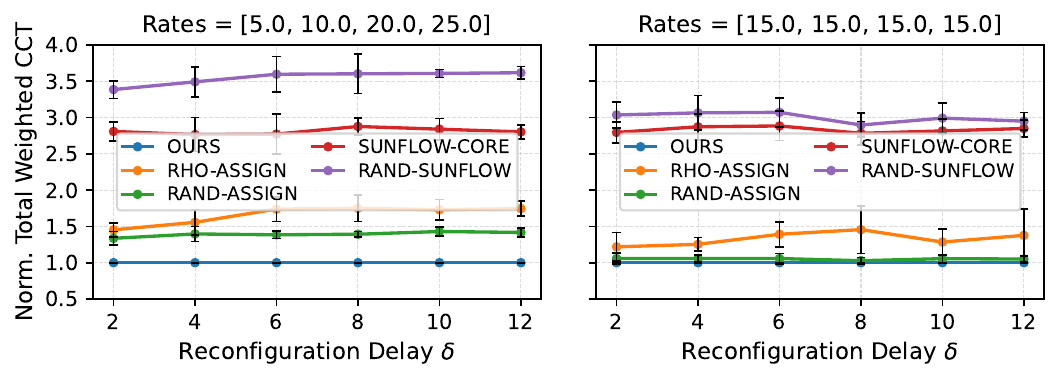}
	\caption{Normalized total weighted CCT versus $\delta$ for $K=4$.}
	\label{fig:delta_k4}
\end{figure}

\begin{itemize}
	\item $K=5$ (Fig.~\ref{fig:delta_k5}).
	Under heterogeneous rates, RHO-ASSIGN and RAND-ASSIGN
	incur approximately $1.42$--$1.73$ and $1.51$--$1.70$
	times the total weighted CCT of $\textsc{Ours}$,
	respectively.
	SUNFLOW-CORE has a larger gap, at approximately
	$2.93$--$3.19$ times, while RAND-SUNFLOW performs
	worst among the compared schemes, at approximately
	$3.90$--$4.39$ times.
	Under homogeneous rates, SUNFLOW-CORE and RAND-SUNFLOW
	still incur approximately $2.90$--$3.17$ and
	$3.14$--$3.29$ times the total weighted CCT of
	$\textsc{Ours}$, respectively.
\end{itemize}

\begin{figure}[h]
	\centering
	\includegraphics[width=9cm,totalheight=8cm,keepaspectratio,height=5.8cm]{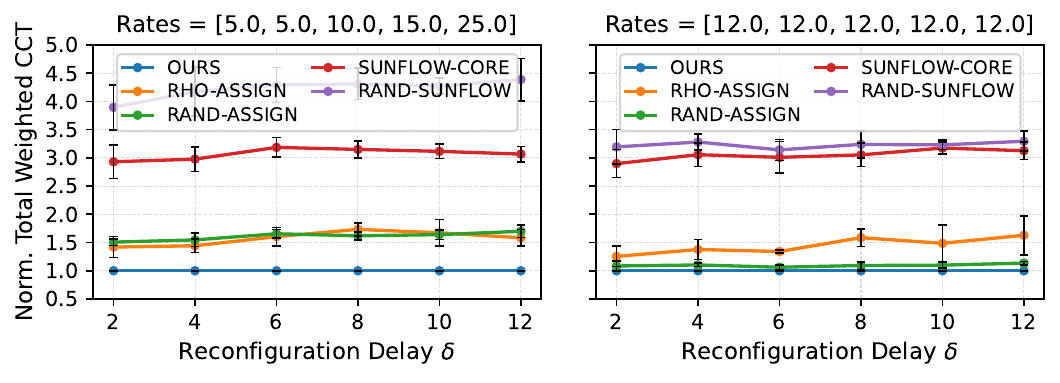}
	\caption{Normalized total weighted CCT versus $\delta$ for $K=5$.}
	\label{fig:delta_k5}
\end{figure}

\subsubsection{Impact of the Number of Ports ($N$-Scaling)}

We fix $M=100$ and $\delta=8$ and vary
$N\in\{8,12,16,24,32\}$.
Tables~\ref{tab:ports_scaling_k3}--\ref{tab:ports_scaling_k5}
report results for $K\in\{3,4,5\}$ under heterogeneous
and homogeneous core rates.
RH-AS, RA-AS, SU-CO, and RA-SU abbreviate RHO-ASSIGN,
RAND-ASSIGN, SUNFLOW-CORE, and RAND-SUNFLOW, respectively.

\begin{table}[htp]
\centering \caption{Normalized total weighted CCT versus $N$ for $K{=}3$.}
\label{tab:ports_scaling_k3} %
\begin{tabular}{ccccc}
\toprule 
$N$ & \textsc{RH-AS} & \textsc{RA-AS} & \textsc{SU-CO} & \textsc{RA-SU}\tabularnewline
\midrule 
\multicolumn{5}{c}{\textbf{Imbalanced rates:} $\left[10,20,30\right]$}\tabularnewline
\midrule 
8 & 1.380 & 1.246 & 2.951 & 3.637\tabularnewline
12 & 1.503 & 1.255 & 2.749 & 3.224\tabularnewline
16 & 1.640 & 1.322 & 2.647 & 3.179\tabularnewline
24 & 1.657 & 1.341 & 2.245 & 2.768\tabularnewline
32 & 1.524 & 1.370 & 2.008 & 2.539\tabularnewline
\midrule 
\multicolumn{5}{c}{\textbf{Balanced rates:} $\left[20,20,20\right]$}\tabularnewline
\midrule 
8 & 1.134 & 1.049 & 3.060 & 3.167\tabularnewline
12 & 1.115 & 1.044 & 2.682 & 2.770\tabularnewline
16 & 1.435 & 1.063 & 2.640 & 2.792\tabularnewline
24 & 1.269 & 1.023 & 2.195 & 2.327\tabularnewline
32 & 1.331 & 1.038 & 2.055 & 2.151\tabularnewline
\bottomrule
\end{tabular}
\end{table}

\begin{table}[htp]
\centering \caption{Normalized total weighted CCT versus $N$ for $K{=}4$.}
\label{tab:ports_scaling_k4} %
\begin{tabular}{ccccc}
\toprule 
$N$ & \textsc{RH-AS} & \textsc{RA-AS} & \textsc{SU-CO} & \textsc{RA-SU}\tabularnewline
\midrule 
\multicolumn{5}{c}{\textbf{Imbalanced rates:} $\left[5,10,20,25\right]$}\tabularnewline
\midrule 
8 & 1.380 & 1.422 & 3.307 & 4.175\tabularnewline
12 & 1.617 & 1.425 & 3.018 & 3.920\tabularnewline
16 & 1.749 & 1.473 & 2.876 & 3.742\tabularnewline
24 & 1.816 & 1.450 & 2.485 & 3.153\tabularnewline
32 & 1.965 & 1.530 & 2.347 & 3.101\tabularnewline
\midrule 
\multicolumn{5}{c}{\textbf{Balanced rates:} $\left[15,15,15,15\right]$}\tabularnewline
\midrule 
8 & 1.163 & 1.099 & 3.208 & 3.365\tabularnewline
12 & 1.322 & 1.060 & 2.970 & 3.147\tabularnewline
16 & 1.456 & 1.008 & 2.783 & 2.913\tabularnewline
24 & 1.634 & 1.059 & 2.513 & 2.699\tabularnewline
32 & 1.601 & 1.056 & 2.350 & 2.469\tabularnewline
\bottomrule
\end{tabular}
\end{table}

\begin{table}[htp]
\centering \caption{Normalized total weighted CCT versus $N$ for $K{=}5$.}
\label{tab:ports_scaling_k5} %
\begin{tabular}{ccccc}
\toprule 
$N$ & \textsc{RH-AS} & \textsc{RA-AS} & \textsc{SU-CO} & \textsc{RA-SU}\tabularnewline
\midrule 
\multicolumn{5}{c}{\textbf{Imbalanced rates:} $\left[5,5,10,15,25\right]$}\tabularnewline
\midrule 
8 & 1.329 & 1.593 & 3.438 & 4.808\tabularnewline
12 & 1.639 & 1.590 & 3.148 & 4.506\tabularnewline
16 & 1.733 & 1.701 & 3.149 & 4.461\tabularnewline
24 & 1.706 & 1.717 & 2.689 & 3.772\tabularnewline
32 & 1.839 & 1.846 & 2.553 & 3.768\tabularnewline
\midrule 
\multicolumn{5}{c}{\textbf{Balanced rates:} $\left[12,12,12,12,12\right]$}\tabularnewline
\midrule 
8 & 1.192 & 1.162 & 3.381 & 3.616\tabularnewline
12 & 1.367 & 1.113 & 3.194 & 3.413\tabularnewline
16 & 1.586 & 1.095 & 3.052 & 3.294\tabularnewline
24 & 1.684 & 1.052 & 2.680 & 2.898\tabularnewline
32 & 1.611 & 1.033 & 2.462 & 2.641\tabularnewline
\bottomrule
\end{tabular}
\end{table}

All reported baseline values exceed one, so $\textsc{Ours}$
achieves the lowest total weighted CCT in every tested
configuration.
The assignment baselines generally have larger relative
gaps under heterogeneous rates, while RAND-ASSIGN remains
close to $\textsc{Ours}$ under homogeneous rates.
The gaps of the two Sunflow-based baselines generally
narrow as $N$ increases, although both remain above
$\textsc{Ours}$ throughout.

\subsubsection{Impact of the Number of Coflows ($M$-Scaling)}

We fix $N=16$ and $\delta=8\,\mathrm{ms}$ and vary
$M\in\{50,100,150,200,250\}$.
Figs.~\ref{fig:coflows50-250}--\ref{fig:coflows50-250-2}
compare heterogeneous and homogeneous core rates
for $K\in\{3,4,5\}$.

\begin{itemize}
	\item $K=3$ (Fig.~\ref{fig:coflows50-250}).
	Under heterogeneous rates, RHO-ASSIGN and RAND-ASSIGN
	have normalized total weighted CCTs around
	$1.29$--$1.64$.
	At larger values of $M$, SUNFLOW-CORE reaches
	approximately $2.65$--$2.75$, while RAND-SUNFLOW
	ranges from approximately $2.76$ to $3.30$ across
	the tested workloads.
	Under homogeneous rates, RHO-ASSIGN and RAND-ASSIGN
	move closer to $\textsc{Ours}$ as $M$ increases,
	reaching approximately $1.10$ and $1.02$ at $M=250$.
	SUNFLOW-CORE and RAND-SUNFLOW retain larger gaps,
	reaching about $2.72$ and $2.85$, respectively.
\end{itemize}

\begin{figure}[h]
	\centering
	\includegraphics[width=9cm,totalheight=8cm,keepaspectratio,height=5.8cm]{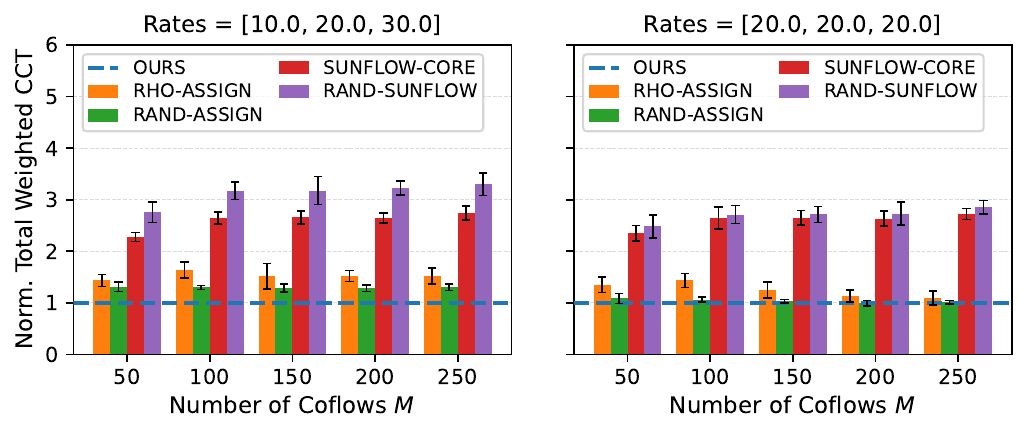}
	\caption{Normalized total weighted CCT versus $M$ for $K=3$.}
	\label{fig:coflows50-250}
\end{figure}

\begin{itemize}
	\item $K=4$ (Fig.~\ref{fig:coflows50-250-1}).
	Under heterogeneous rates, RHO-ASSIGN and RAND-ASSIGN
	have normalized values of approximately $1.69$--$1.83$
	and $1.38$--$1.43$, respectively.
	The relative gaps of SUNFLOW-CORE and RAND-SUNFLOW
	increase with $M$, reaching about $3.08$ and $3.91$
	at $M=250$.
	Under homogeneous rates, RHO-ASSIGN falls to about
	$1.23$ at $M=250$.
	For the larger values of $M$, RAND-ASSIGN is comparable
	to $\textsc{Ours}$, with normalized values of approximately
	$0.97$--$1.05$.
	SUNFLOW-CORE and RAND-SUNFLOW remain higher,
	at approximately $2.53$--$3.09$.
\end{itemize}

\begin{figure}[h]
	\centering
	\includegraphics[width=9cm,totalheight=8cm,keepaspectratio,height=5.8cm]{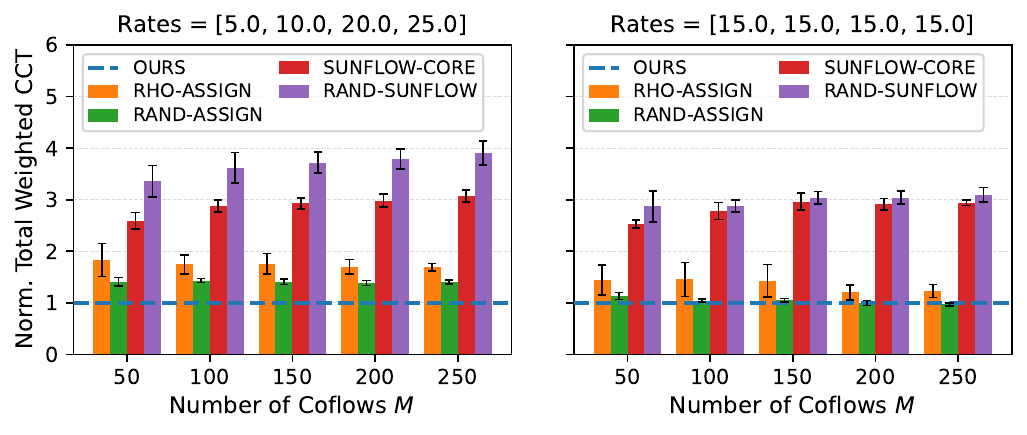}
	\caption{Normalized total weighted CCT versus $M$ for $K=4$.}
	\label{fig:coflows50-250-1}
\end{figure}

\begin{itemize}
	\item $K=5$ (Fig.~\ref{fig:coflows50-250-2}).
	Under heterogeneous rates, the relative gaps of the
	two Sunflow-based baselines increase with $M$.
	SUNFLOW-CORE rises from approximately $2.85$ to $3.25$,
	and RAND-SUNFLOW from approximately $3.93$ to $4.57$.
	RHO-ASSIGN and RAND-ASSIGN remain around
	$1.46$--$1.75$ and $1.59$--$1.68$, respectively.
	Under homogeneous rates, the two assignment baselines
	move closer to $\textsc{Ours}$, reaching about $1.18$
	and $1.02$, respectively, at $M=250$.
	At the same workload size, SUNFLOW-CORE and RAND-SUNFLOW
	reach approximately $3.29$ and $3.46$.
\end{itemize}

\begin{figure}[h]
	\centering
	\includegraphics[width=9cm,totalheight=8cm,keepaspectratio,height=5.8cm]{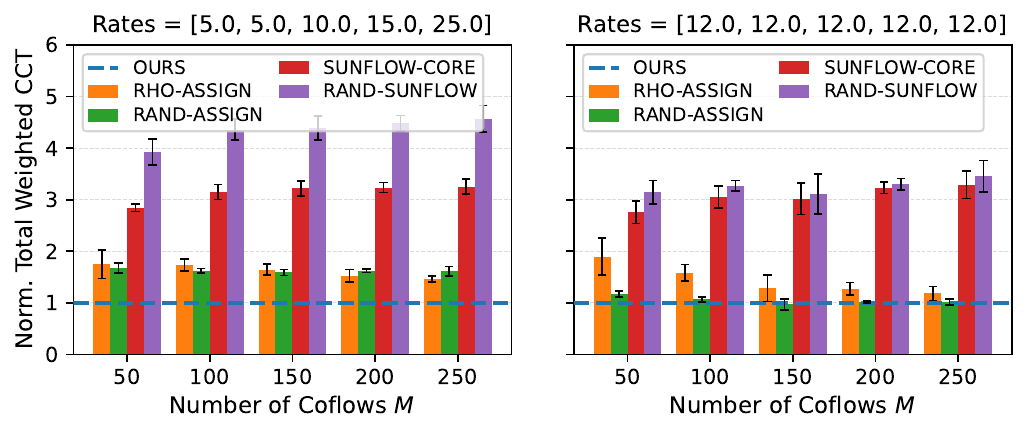}
	\caption{Normalized total weighted CCT versus $M$ for $K=5$.}
	\label{fig:coflows50-250-2}
\end{figure}

Overall, the relative advantage of $\textsc{Ours}$ is
more pronounced under heterogeneous rates.
Rate-proportional random assignment can be competitive
under homogeneous rates, particularly at larger workloads.

\section{Conclusions}
\label{sec:Conclusions}

We develop an approximation algorithm for coflow scheduling
in multi-core OCS networks under not-all-stop reconfiguration.
By coordinating flow assignment and per-core circuit scheduling,
the algorithm achieves a
$2M\frac{w_{\max}}{w_{\min}}\max\{K,\tau_{\max}\}$-approximation
for minimizing total weighted CCT.
We also derive a refined bound based on coflow weights and
workloads.
For multi-core EPS networks, setting the reconfiguration delay
to zero and adapting the lower bounds yields a
$2M\frac{w_{\max}}{w_{\min}}H$-approximation.
Trace-driven experiments support the algorithm's effectiveness
in the evaluated settings.

Future work will focus on online scheduling with incomplete
demand information and competitive guarantees that account
for circuit reconfiguration costs.

\section*{Acknowledgement}

This work is supported by the Department of Environment, Science and
Innovation of Queensland State Government under Quantum 2032 Challenge
Program (Project \#Q2032001).

\bibliographystyle{IEEEtran}
\addcontentsline{toc}{section}{\refname}\bibliography{IEEE_Journal}

@article{Baraat,
  author    = {Dogar, Fahad R and Karagiannis, Thomas and Ballani, Hitesh and Rowstron, Antony},
  journal   = {ACM SIGCOMM Computer Communication Review},
  title     = {Decentralized task-aware scheduling for data center networks},
  year      = {2014},
  number    = {4},
  pages     = {431--442},
  volume    = {44},
  publisher = {ACM New York, NY, USA},
}

@article{Aalo,
  author    = {Chowdhury, Mosharaf and Stoica, Ion},
  journal   = {ACM SIGCOMM Computer Communication Review},
  title     = {Efficient coflow scheduling without prior knowledge},
  year      = {2015},
  number    = {4},
  pages     = {393--406},
  volume    = {45},
  publisher = {ACM New York, NY, USA},
}

@misc{facebook,
  howpublished = {\url{https://github.com/coflow/coflow-benchmark}},
  title        = {FaceBookTrace},
  year         = {2019},
}

@article{mapreduce,
  author    = {Dean, Jeffrey and Ghemawat, Sanjay},
  journal   = {Communications of the ACM},
  title     = {MapReduce: simplified data processing on large clusters},
  year      = {2008},
  number    = {1},
  pages     = {107--113},
  volume    = {51},
  publisher = {ACM New York, NY, USA},
}

@inProceedings{spark,
  author    = {Zaharia, Matei and Chowdhury, Mosharaf and Das, Tathagata and Dave, Ankur and Ma, Justin and McCauly, Murphy and Franklin, Michael J and Shenker, Scott and Stoica, Ion},
  booktitle = {9th $\{$USENIX$\}$ Symposium on Networked Systems Design and Implementation ($\{$NSDI$\}$ 12)},
  title     = {Resilient distributed datasets: A fault-tolerant abstraction for in-memory cluster computing},
  year      = {2012},
  pages     = {15--28},
}

@inProceedings{dryad,
  author    = {Isard, Michael and Budiu, Mihai and Yu, Yuan and Birrell, Andrew and Fetterly, Dennis},
  booktitle = {Proceedings of the 2nd ACM SIGOPS/EuroSys European Conference on Computer Systems 2007},
  title     = {Dryad: distributed data-parallel programs from sequential building blocks},
  year      = {2007},
  pages     = {59--72},
}

@inProceedings{networking,
  author    = {Chowdhury, Mosharaf and Stoica, Ion},
  booktitle = {Proceedings of the 11th ACM Workshop on Hot Topics in Networks},
  title     = {Coflow: A networking abstraction for cluster applications},
  year      = {2012},
  pages     = {31--36},
}

@inProceedings{literature9,
	title={Minimizing the total weighted completion time of coflows in datacenter networks},
	author={Qiu, Zhen and Stein, Cliff and Zhong, Yuan},
	booktitle={Proceedings of the 27th ACM symposium on Parallelism in Algorithms and Architectures},
	pages={294--303},
	year={2015}
}

@article{literature4,
	title={Managing data transfers in computer clusters with orchestra},
	author={Chowdhury, Mosharaf and Zaharia, Matei and Ma, Justin and Jordan, Michael I and Stoica, Ion},
	journal={ACM SIGCOMM Computer Communication Review},
	volume={41},
	number={4},
	pages={98--109},
	year={2011},
	publisher={ACM New York, NY, USA}
}

@inProceedings{literature6,
	title={Efficient coflow scheduling with varys},
	author={Chowdhury, Mosharaf and Zhong, Yuan and Stoica, Ion},
	booktitle={Proceedings of the 2014 ACM conference on SIGCOMM},
	pages={443--454},
	year={2014}
}

@inProceedings{literature10,
	title={Brief announcement: Improved approximation algorithms for scheduling co-flows},
	author={Khuller, Samir and Purohit, Manish},
	booktitle={Proceedings of the 28th ACM Symposium on Parallelism in Algorithms and Architectures},
	pages={239--240},
	year={2016}
}

@inProceedings{sunflow,
	author    = {Huang, Xin Sunny and Sun, Xiaoye Steven and Ng, T. S. Eugene},
	title     = {Sunflow: Efficient optical circuit scheduling for coflows},
	booktitle = {Proceedings of the 12th International Conference on Emerging Networking Experiments and Technologies},
	pages     = {297--311},
	year      = {2016},
}

@inProceedings{co-scheduler,
  author       = {Li, Zhuozhao and Shen, Haiying},
  booktitle    = {2019 IEEE 39th International Conference on Distributed Computing Systems (ICDCS)},
  title        = {Co-scheduler: Accelerating data-parallel jobs in datacenter networks with optical circuit switching},
  year         = {2019},
  organization = {IEEE},
  pages        = {186--195},
}

@article{zhang2020minimizing,
	title={Minimizing coflow completion time in optical circuit switched networks},
	author={Zhang, Tong and Ren, Fengyuan and Bao, Jiakun and Shu, Ran and Cheng, Wenxue},
	journal={IEEE Transactions on Parallel and Distributed Systems},
	volume={32},
	number={2},
	pages={457--469},
	year={2020},
	publisher={IEEE}
}

@inProceedings{omco,
  author       = {Xu, Chao and Tan, Haisheng and Hou, Jiahui and Zhang, Chi and Li, Xiang-Yang},
  booktitle    = {2018 IEEE International Conference on Communications (ICC)},
  title        = {OMCO: Online multiple coflow scheduling in optical circuit switch},
  year         = {2018},
  organization = {IEEE},
  pages        = {1--6},
}

@article{BvN,
  author  = {Birkhoff, Garrett},
  journal = {Univ. Nac. Tucuman, Ser. A},
  title   = {Tres observaciones sobre el algebra lineal},
  year    = {1946},
  pages   = {147--154},
  volume  = {5},
}

@article{regularization,
  author    = {Tan, Haisheng and Zhang, Chi and Xu, Chao and Li, Yupeng and Han, Zhenhua and Li, Xiang-Yang},
  journal   = {IEEE/ACM Transactions on Networking},
  title     = {Regularization-based coflow scheduling in optical circuit switches},
  year      = {2021},
  number    = {3},
  pages     = {1280--1293},
  volume    = {29},
  publisher = {IEEE},
}

@article{improved,
  author={Shafiee, Mehrnoosh and Ghaderi, Javad},
  journal={IEEE/ACM Transactions on Networking},
  title={An improved bound for minimizing the total weighted completion time of coflows in datacenters},
  volume={26},
  number={4},
  pages={1674--1687},
  year={2018},
  publisher={IEEE}
}

@inProceedings{theoretical1,
  author       = {Zhao, Yangming and Chen, Kai and Bai, Wei and Yu, Minlan and Tian, Chen and Geng, Yanhui and Zhang, Yiming and Li, Dan and Wang, Sheng},
  booktitle    = {2015 IEEE Conference on Computer Communications (INFOCOM)},
  title        = {Rapier: Integrating routing and scheduling for coflow-aware data center networks},
  year         = {2015},
  organization = {IEEE},
  pages        = {424--432},
}

@inproceedings{theoretical4,
	author    = {Im, Sungjin and Moseley, Benjamin and Pruhs, Kirk and Purohit, Manish},
	title     = {Matroid Coflow Scheduling},
	booktitle = {46th International Colloquium on Automata, Languages, and Programming (ICALP)},
	volume    = {132},
	pages     = {145:1--145:13},
	year      = {2019},
	doi       = {10.4230/LIPIcs.ICALP.2019.145},
}

@article{theoretical5,
	title={Scheduling coflows with dependency graph},
	author={Shafiee, Mehrnoosh and Ghaderi, Javad},
	journal={IEEE/ACM Transactions on Networking},
	volume={30},
	number={1},
	pages={450--463},
	year={2021},
	publisher={IEEE}
}

@inProceedings{decentralized1,
  author       = {Luo, Shouxi and Yu, Hongfang and Zhao, Yangming and Wu, Bin and Wang, Sheng and others},
  booktitle    = {2015 IEEE International Conference on Communications (ICC)},
  title        = {Minimizing average coflow completion time with decentralized scheduling},
  year         = {2015},
  organization = {IEEE},
  pages        = {307--312},
}

@inProceedings{ONS,
  author       = {Jiang, Renjie and Zhang, Tong and Yi, Changyan},
  booktitle    = {2023 IEEE Symposium on Computers and Communications (ISCC)},
  title        = {Effective Coflow Scheduling in Hybrid Circuit and Packet Switching Networks},
  year         = {2023},
  organization = {IEEE},
  pages        = {1156--1161},
}

@inProceedings{reco,
  author       = {Zhang, Chi and Tan, Haisheng and Xu, Chao and Li, Xiang-Yang and Tang, Shaojie and Li, Yupeng},
  booktitle    = {2019 IEEE 39th International Conference on Distributed Computing Systems (ICDCS)},
  title        = {Reco: Efficient regularization-based coflow scheduling in optical circuit switches},
  year         = {2019},
  organization = {IEEE},
  pages        = {111--121},
}

@article{wang2023efficient,
  author    = {Wang, Xin and Shen, Hong and Tian, Hui},
  journal   = {IEEE Transactions on Network and Service Management},
  title     = {Efficient and Fair: Information-Agnostic Online Coflow Scheduling by Combining Limited Multiplexing with DRL},
  year      = {2023},
  number    = {4},
  pages     = {4572--4584},
  volume    = {20},
  publisher = {IEEE},
}

@article{wang2023online,
  author    = {Wang, Xin and Shen, Hong},
  journal   = {Future Generation Computer Systems},
  title     = {Online scheduling of coflows by attention-empowered scalable deep reinforcement learning},
  year      = {2023},
  pages     = {195--206},
  volume    = {146},
  publisher = {Elsevier},
}

@article{wang2024scheduling,
  title={Scheduling coflows in hybrid optical-circuit and electrical-packet switches with performance guarantee},
  author={Wang, Xin and Shen, Hong and Tian, Hui},
  journal={IEEE/ACM Transactions on Networking},
  volume={32},
  number={3},
  pages={2299--2314},
  year={2024},
  publisher={IEEE}
}

@article{chen2023scheduling,
  title={Scheduling coflows for minimizing the total weighted completion time in heterogeneous parallel networks},
  author={Chen, Chi-Yeh},
  journal={Journal of Parallel and Distributed Computing},
  volume={182},
  pages={104752},
  year={2023},
  publisher={Elsevier}
}

@article{iaoa,
  author    = {Wang, Zhiliang and Zhang, Han and Shi, Xingang and Yin, Xia and Li, Yahui and Geng, Haijun and Wu, Qianhong and Liu, Jianwei},
  journal   = {IEEE Transactions on Parallel and Distributed Systems},
  title     = {Efficient scheduling of weighted coflows in data centers},
  year      = {2019},
  number    = {9},
  pages     = {2003--2017},
  volume    = {30},
  publisher = {IEEE},
}

@inProceedings{CODA,
  author    = {Zhang, Hong and Chen, Li and Yi, Bairen and Chen, Kai and Chowdhury, Mosharaf and Geng, Yanhui},
  booktitle = {Proceedings of the 2016 ACM SIGCOMM Conference},
  title     = {Coda: Toward automatically identifying and scheduling coflows in the dark},
  year      = {2016},
  pages     = {160--173},
}

@inProceedings{huang2020weaver,
  title={Weaver: Efficient coflow scheduling in heterogeneous parallel networks},
  author={Huang, Xin Sunny and Xia, Yiting and Ng, TS Eugene},
  booktitle={2020 IEEE International Parallel and Distributed Processing Symposium (IPDPS)},
  pages={1071--1081},
  year={2020},
  organization={IEEE}
}

@article{wang2025optimal,
	author    = {Wang, Xin and Shen, Hong and Tian, Hui},
	title     = {Optimal Partitioning of Traffic Demand for Coflow Scheduling in Hybrid Switches},
	journal   = {IEEE Transactions on Network and Service Management},
	volume    = {22},
	number    = {6},
	pages     = {5972--5986},
	year      = {2025},
	publisher = {IEEE},
}

@misc{cisco2016gci,
  author       = {{Cisco. (2016)}},
  title        = {Cisco Global Cloud Index: Forecast and Methodology, 2015--2020},
  howpublished = {\url{https://www.cisco.com/c/dam/en/us/solutions/collateral/service-provider/global-cloud-index-gci/white-paper-c11-738085.pdf}},
}

@misc{cisco2016_40g,
  author       = {{Cisco White Paper. (2016)}},
  title        = {The Future is 40 Gigabit Ethernet},
  howpublished = {\url{https://www.cisco.com/c/dam/en/us/products/collateral/switches/catalyst-6500-series-switches/white-paper-c11-737238.pdf}},
}

@article{chen2023efficient,
  title={Efficient approximation algorithms for scheduling coflows with total weighted completion time in identical parallel networks},
  author={Chen, Chi-Yeh},
  journal={IEEE Transactions on Cloud Computing},
  volume={12},
  number={1},
  pages={116--129},
  year={2023},
  publisher={IEEE}
}

@inProceedings{poutievski2022jupiter,
  title={Jupiter evolving: transforming google's datacenter network via optical circuit switches and software-defined networking},
  author={Poutievski, Leon and Mashayekhi, Omid and Ong, Joon and Singh, Arjun and Tariq, Mukarram and Wang, Rui and Zhang, Jianan and Beauregard, Virginia and Conner, Patrick and Gribble, Steve and others},
  booktitle={Proceedings of the ACM SIGCOMM 2022 Conference},
  pages={66--85},
  year={2022}
}

@inproceedings{wang2026coflow,
	title={Coflow Scheduling in Hybrid-Switched Data Center Networks under Not-All-Stop Reconfiguration},
	author={Wang, Xin and Shen, Hong and Tian, Hui},
	booktitle={Proceedings of the 55th International Conference on Parallel Processing},
	pages={1008--1018},
	year={2026}
}

\vspace{12pt}

\end{document}